\pdfoutput=1
\documentclass[12pt]{article}

\usepackage{amsmath}
\usepackage{graphicx}
\usepackage{enumerate}
\usepackage{natbib}
\usepackage{comment, setspace, float}
\usepackage{url} %
\RequirePackage{amsthm,amsmath,amsfonts,amssymb,bbm}
\RequirePackage{graphicx}
\usepackage[ruled]{algorithm2e}
\usepackage{multirow}

\theoremstyle{plain}

\newtheorem{theorem}{Theorem}[section]
\newtheorem{lemma}[theorem]{Lemma}
\theoremstyle{remark}

\newtheorem{corollary}[theorem]{Corollary}

\allowdisplaybreaks
\usepackage{mathtools}
\usepackage{caption}
\usepackage{subcaption}

\DeclareMathOperator{\Unif}{Unif}

\DeclareMathOperator{\tr}{tr}

\newcommand\bx{\mathbf{x}}
\newcommand{\s}{s}

\newcommand{\m}{\mid}
\newcommand{\abs}[1]{\left\vert #1 \right\vert}

\newcommand \bbE{\mathbb{E}}

\newcommand \bbR{\mathbb{R}}

\newcommand \bbC{\mathbb{C}}

\newcommand{\lvertiii}{{\vert\kern-0.25ex\vert\kern-0.25ex\vert}}    
\newcommand{\rvertiii}{{\vert\kern-0.25ex\vert\kern-0.25ex\vert}}

\renewcommand{\tr}{\mathrm{tr}}

\newcommand{\N}{N}
\newcommand{\K}{K}

\newcommand{\mc}[1]{\mathcal #1}

\newcommand{\set}[1]{\left\{#1\right\}}
\newcommand{\norm}[1]{\left\|#1\right\|_{2}}

\newcommand{\p}[1]{\left(#1\right)}

\newcommand{\be}{\begin{equs}}
\newcommand{\ee}{\end{equs}}

\newcommand{\bigO}{\mathcal O}

\def\T{{ \mathrm{\scriptscriptstyle T} }}
\def \bX{\mathbf{X}}
\def \œmb{\mathbb}
\def \mc{\mathcal}
\def \mrm{\mathrm}

\usepackage[margin = 0.7in]{geometry}

\newcommand{\cN}{\mathcal{N}}          %
\newcommand{\cY}{\mathcal{Y}}          %
\newcommand{\cD}{\mathcal{D}}          %
\newcommand{\cL}{\mathcal{L}}          %
\newcommand{\iter}{\text{iter}}       %
  
\newcommand{\lpart}{\psi}              %
\newcommand{\Hel}{\mathrm{H}}          %
\newcommand{\dTheta}{d_{\Theta}}       %
\newcommand{\ourmodel}{ExFR-LGP}   %
\newcommand{\ourgmodel}{Ex2FR-LGP}   %

\newcommand{\blind}{0}

\usepackage{titlesec}
\usepackage{hyperref}
\usepackage{xcolor} 
\usepackage{authblk}

\begin{document}

\def\spacingset#1{\renewcommand{\baselinestretch}%
{#1}\small\normalsize} \spacingset{1}

\if0\blind
{
  \title{\bf Finite-Rank Logistic Gaussian Processes with Exact Likelihood for Conditional Density Estimation}
  \small
  \author[1]{Jaehoan Kim\thanks{Equal contribution}}
  \author[2]{Indrajit Ghosh\protect\footnotemark[1]}
  \author[3]{Debdeep Pati}
  \author[4]{Dipankar Bandyopadhyay\thanks{Data used in preparation of this article were obtained from the Alzheimer’s Disease Neuroimaging Initiative (ADNI) database (\url{https://adni.loni.usc.edu/}). %
    The authors acknowledge funding support from grants R21DE031879 and R01DE031134 from the United States National Institutes of Health. \\ Correspondence E-mail: \texttt{dbandyop@vcu.edu}}}
  \affil[1]{Department of Statistical Science, Duke University, Durham, NC}
  \affil[2]{Global Statistical Sciences, Eli Lilly and Company, Indianapolis, IN}
  \affil[3]{Department of Statistics, University of Wisconsin-Madison, Madison, WI}
  \affil[4]{Department of Biostatistics, Virginia Commonwealth University, Richmond, VA}
\maketitle
}\fi
\if1\blind
{
  \bigskip
  \bigskip
  \bigskip

\begin{center}
    {\LARGE\bf Finite-Rank Logistic Gaussian Processes with Exact Likelihood for Conditional Density Estimation}
\end{center}
  \medskip
} \fi

\bigskip
\begin{abstract}
Conditional density estimation describes how the entire distribution of a response changes with covariates, and in imaging studies also with location. Logistic Gaussian processes (LGP) give a flexible prior for such densities. However, the normalizing constant of an LGP has no closed form, so existing methods approximate or replace the likelihood. We propose the exact likelihood finite-rank LGP (\ourmodel{}), which writes the log density as the sum of two bivariate functions, one of the response and a location-varying linear index of the covariates, and one of the response and the location. Each function is assigned a finite-rank Gaussian process prior that is piecewise linear on a regular grid, enabling the normalizing constant, the conditional mean and the quantiles to enjoy closed form representations. Posterior samples are drawn from the exact posterior under this prior via a Gibbs sampler that updates the Gaussian components by elliptical slice sampling. When the true log density is the sum of two such bivariate functions, we show that the posterior contracts at the minimax rate of an $\alpha$-smooth bivariate density up to a logarithmic factor. Numerical experiments using synthetic data demonstrate the effectiveness of \ourmodel{} in conditional density estimation. Furthermore, application to fractional anisotropy responses along the corpus callosum in the Alzheimer's Disease Neuroimaging Initiative data yields covariate-adjusted percentile bands with uncertainty, elucidating how diagnosis changes the distribution along the tract.
\end{abstract}

\noindent%
\vfill

\newpage
\spacingset{1.9} %
\section{Introduction}\label{sec: introduction}

In diverse fields such as biology \citep{beaumont2002approximate}, cosmology \citep{rau2015accurate}, ecology \citep{wood2010statistical}, and economics \citep{kalda2013nonparametric}, there is a pressing need for modeling the distribution of a dependent variable given the covariates, rather than merely modeling conditional mean or variance. Such problems are referred to as conditional density estimation (CDE), where given a response-covariate pair $(y,\, \bx) \in\, \bbR\, \times\, \bbR^p$, one aims to find the density function of $y$ conditioned on $\bx$, denoted by $f(y \m \bx)$.

As a motivating example, we consider the data from the Alzheimer’s Disease Neuroimaging Initiative (ADNI) study \citep{jack2008alzheimer}. For each individual, a summary measure called the Fractional Anisotropy (FA) measuring the brain connectivity \citep{ozarslan2005generalized} is calculated at varying locations across the skeleton of the midsagittal corpus callosum, along with different features including age, sex, and handedness. As demonstrated in Figures \ref{fig:box_res} and \ref{fig:matplot_res}, the FA measurements vary not only with respect to individuals, but also according to measurement locations. Based on the FA measurements, clinically relevant questions include where along the corpus callosum an individual’s FA values fall outside a covariate-adjusted reference percentile band, and how extensive these departures are. Therefore, our main interest is to model an individual's FA measurement through the conditional density $f(y\mid\bx,\s)$ of the FA value $y$ given the covariates $\bx$ and the location $\s$ along the tract. The fitted density should support uncertainty quantification for clinically relevant summaries, such as covariate-adjusted percentile curves, and should also indicate how multiple covariates and location alter the FA distribution.

 \begin{figure}[t!]
    \centering
\begin{subfigure}{0.50\textwidth}
  \centering
  \includegraphics[scale= 0.215]{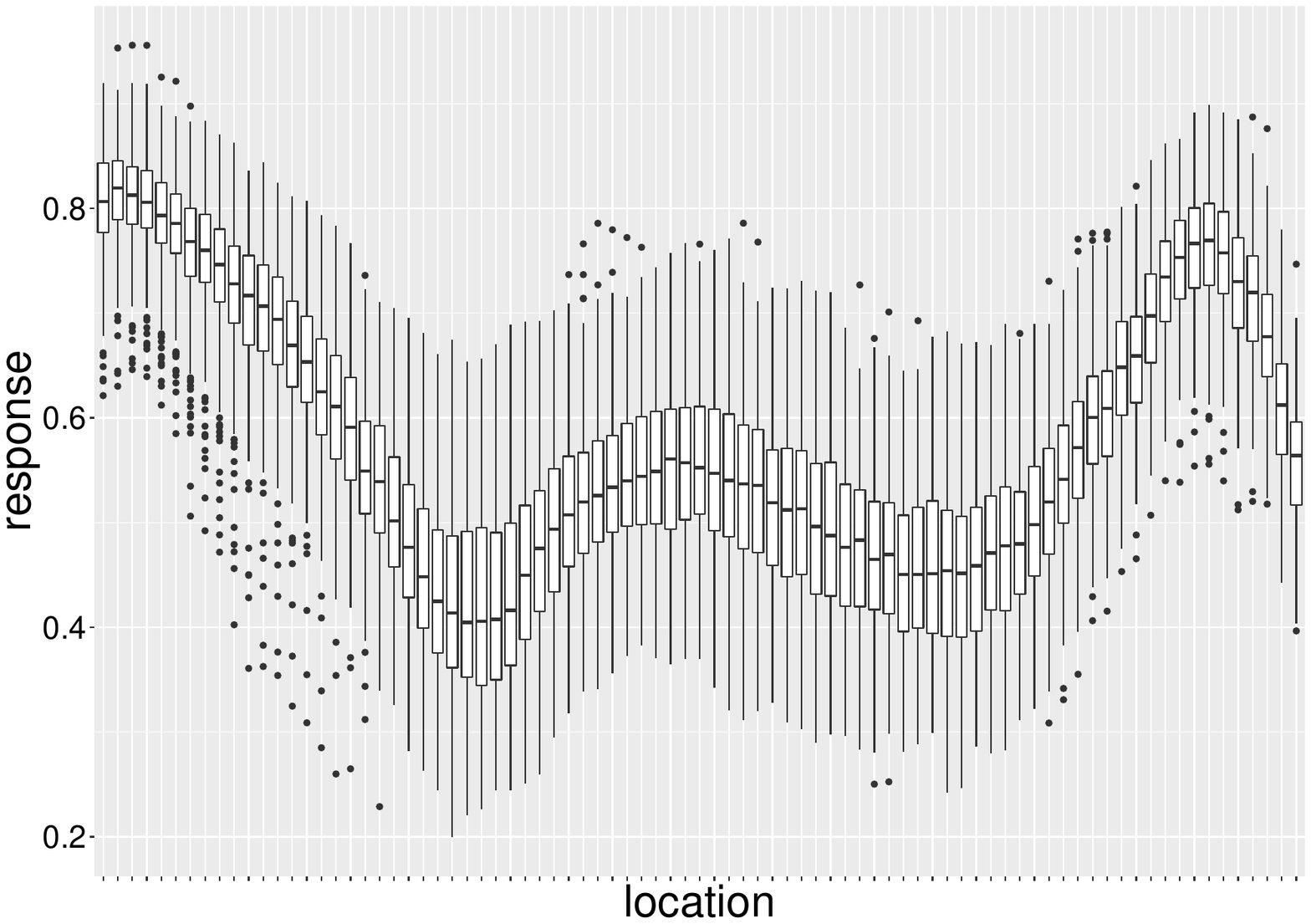}
  \caption{}
  \label{fig:box_res}
\end{subfigure}%
\begin{subfigure}{0.50\textwidth}
  \centering
  \includegraphics[scale= 0.215]{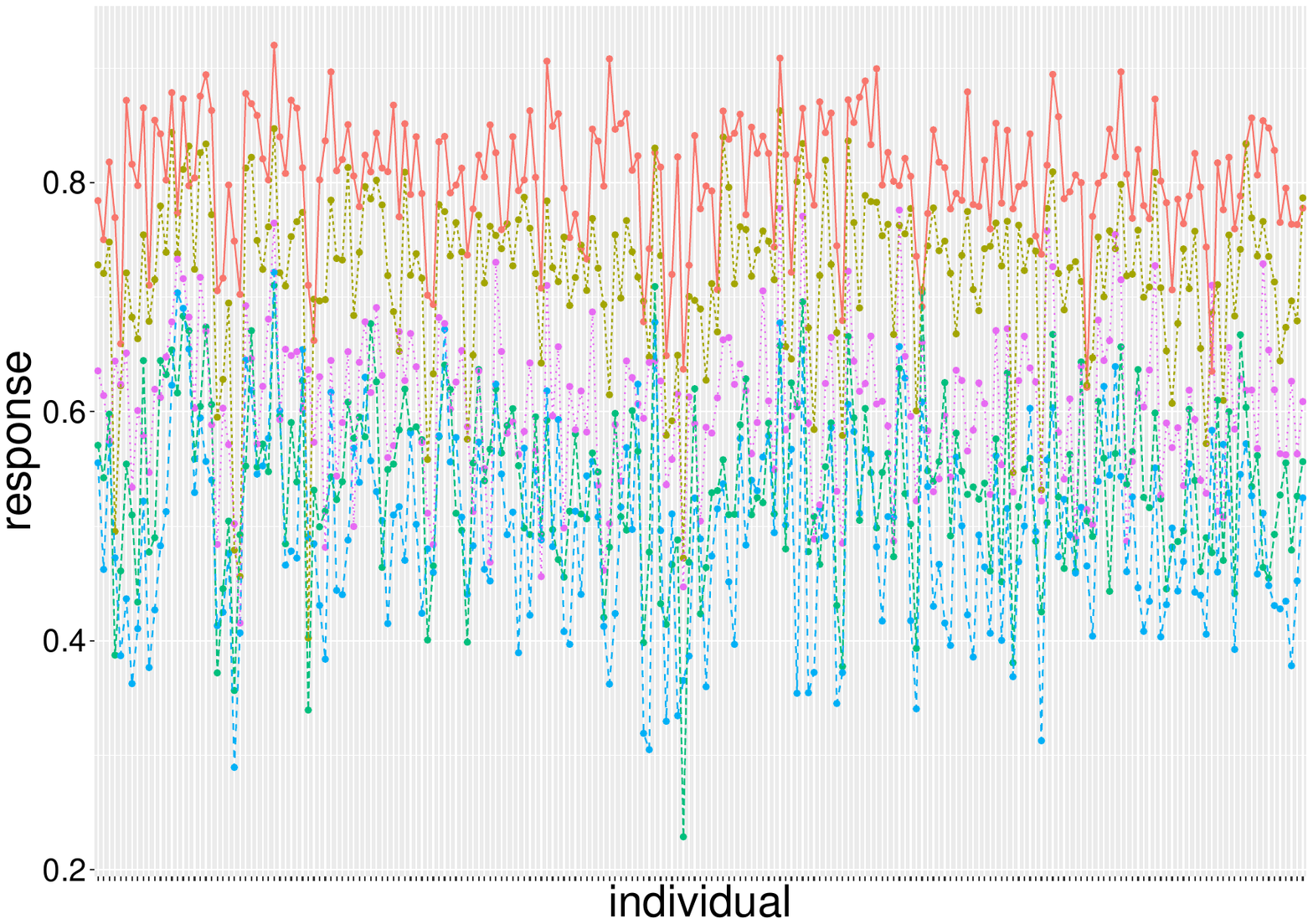}
  \caption{}
  \label{fig:matplot_res}
\end{subfigure}%
\caption{Exploratory analysis of the ADNI data. The left panel (Figure 1a) depicts the boxplots of the responses over $83$ different locations, while the right panel (Figure 1b) plots the FA responses for $213$ individuals from $5$ different locations (marked by varying colors).}
\label{fig:EDA}
\end{figure}

Quantile regression targets such percentile bands directly, and \cite{zhang2021high} fit the quantiles of FA along the same corpus callosum skeleton in ADNI \citep[see also][]{reich2011bayesian}. Each quantile is then linear in the covariates, and other distributional properties such as density, modes, and tail probabilities must be derived from the fitted quantile curves. To address this, one possible approach is parametric distributional regression \citep{rigby2005generalized, klein2015bayesian}. However, as it fixes a distributional family, the shape of the density including its number of modes is limited by that family, which might inhibit flexible density modeling.

CDE addresses this problem by nonparametrically modeling the density function so that the relevant properties such as tail probabilities and modes can be obtained directly from the density estimate. CDE has a well-established frequentist literature, including kernel \citep{chen2017tutorial} and local polynomial \citep{fan1996estimation} methods, alongside a growing body of Bayesian approaches \citep{mcdonald2021review}. One Bayesian approach models the joint density of $(y,\bx)$ and obtains the conditional density from it. This approach requires modeling the covariate distribution, including the four binary covariates in our application (sex, handedness, AD status, and MCI status). Moreover, when several covariates are included, clustering may be driven primarily by $\bx$ rather than by $y$ \citep{wade2014improving}. The other route models $f(y\mid\bx)$ directly, mostly as a mixture, whose weights or atoms depend on the covariates \citep{dunson2007bayesian, griffin2006order, dunson2008kernel}, with spatial versions \citep{gelfand2005bayesian, reich2007multivariate} and, for responses on $[0,1]$, dependent Bernstein polynomial priors \citep{petrone1999random, barrientos2017fully}. Mixtures are designed for heterogeneity that arises from latent subpopulations. In the FA data, the main sources of heterogeneity are observed covariates and the density changes gradually along the tract (Figure \ref{fig:box_res}). In a mixture, such a change is spread over the weights and atoms, so that the effect of a single covariate on it does not correspond to a parameter of the model.

For this purpose, the logistic Gaussian process \citep[LGP;][]{leonard1978density, lenk1988logistic} has been widely studied. It writes the density as $\exp\set{Z(y)}$ divided by its integral over the response space, where $Z$ is modeled as a Gaussian process (GP). Since its exponential integration does not have a closed form, existing LGP methods largely fall under two classes. The first discretizes the GP and computes the integral numerically. It tracks the process at a finite set of nodes \citep{tokdar2007towards, tokdar2010bayesian}, bins the data with a Laplace approximation \citep{riihimaki2014laplace, payne2020conditional}, or represents the GP by random Fourier features \citep{gautier2025continuous} or by piecewise multilinear finite elements on a uniform grid \citep{hegland2007approximate, griebel2010finite}. For example, \cite{jacobs2014nonparametric} discretizes an LGP on a Riemannian manifold. The second class avoids the normalizing constant. \cite{adams2009gaussian} pass the GP through a bounded link and treat the rejected proposals of a rejection sampler as latent variables, with extensions to Poisson process intensities \citep{adams2009tractable} and P\'{o}lya-Gamma augmentation \citep{donner2018efficient, wang2024nonparametric}. Others replace the likelihood by a score matching loss \citep{chen2026logistic, paisley2025gaussian}.

Such LGP methods, however, face two limitations in our modeling context. The first class computes the integral numerically even after discretizing the GP, so the resulting likelihood is an approximation. The second class changes the link and adds latent points whose number grows with $n$, or replaces the likelihood. In each case, moments and quantiles of the fitted density need numerical integration. Prediction-wise, for estimating a conditional density given a $p$-dimensional covariate vector $\bx$, the minimax rate for an $\alpha$-smooth conditional density is $n^{-\alpha/(2\alpha+p+1)}$, which deteriorates quickly as $p$ grows \citep{tsybakov2004introduction}. \cite{tokdar2010bayesian} reduce dimensionality by projecting $\bx$ onto a low-dimensional subspace, yielding faster posterior contraction when the true conditional density depends on $\bx$ only through such a (low-rank) projection \citep{tokdar2011dimension}, while \cite{odin2026contraction} extend such rates for regression and density estimation to dimensions that grow with $n$. Such projection LGP methods, however, still approximate the normalizing constant.

In this paper, we propose the exact likelihood finite-rank LGP (\ourmodel{}), a finite-rank LGP model whose likelihood is exact. We model the conditional density using two bivariate functions: $g_1$ captures covariate effects through a linear index with location-dependent coefficients, while $g_2$ captures additional variation across locations:
\begin{align}
    f(y\mid \bx,\s) = \frac{\exp\set{g_1(y,\bx^\T\beta(\s)) + g_2(y,\s)}}{\int_{\cY} \exp\set{g_1(t,\bx^\T\beta(\s)) + g_2(t,\s)}dt},
\end{align}
where $g_1$ and $g_2$ are assigned bivariate GP priors, and $\beta(\s)$ is a $p$-dimensional vector of coefficients varying with location $\s$, with $\cY=[0,1]$ in our setting. Similar to single index models \citep{li1991sliced} and generalized additive models \citep{hastie1986generalized}, each component is easy to read. We represent $g_1$ and $g_2$ by linear interpolation of their values on a regular grid \citep{maatouk2017gaussian}, the tensor-product finite elements of \cite{griebel2010finite}, which here define the prior rather than approximate the posterior mode. The interpolated process is a finite-rank approximation of the parent GP; so the model is a finite-rank LGP, and for fixed $(\bx,\s)$ the exponent is piecewise linear in $y$. Hence, the normalizing constant, the mean and the quantiles are available in closed form (Supplement \ref{ssec: normalizing constant}, \ref{ssec: mean} and \ref{ssec: quantiles}), and the sampler targets the exact posterior under this prior. To our knowledge, this is the \textit{first} finite-rank LGP model for conditional density estimation with a closed-form likelihood, requiring neither response binning nor numerical quadrature and avoiding the use of a surrogate loss. Since $g_1$ and $g_2$ are bivariate, each grid stays two dimensional for any $p$. When the true log density is the sum of two such bivariate functions, we show that the posterior contracts at the rate $n^{-\alpha/(2\alpha+2)}$ up to a log factor, which is the minimax rate for an $\alpha$-smooth bivariate density and the rate of \cite{tokdar2011dimension} under a rank one projection. Since the grid values have a Gaussian prior induced by the GP, we sample them by the elliptical slice sampler \citep[ESS;][]{murray2010elliptical}, which requires no tuning parameters. Finally, we investigate the performance of \ourmodel{} in simulations and on the ADNI data. 

The paper is organized as follows. Section~\ref{sec: methodology} defines \ourmodel{} and its prior, discusses its identifiability and alternative models. Section~\ref{sec: posterior contraction} presents the posterior contraction rate, and Section~\ref{sec: posterior computation} derives the sampler. While Section~\ref{sec: simulation} reports finite-sample characteristics of our proposal via simulation studies, Section~\ref{sec: real data} illustrates \ourmodel{} via application to the ADNI data, presenting covariate-adjusted percentile curves along the corpus callosum with uncertainty. Finally, Section~\ref{sec: conclusion} concludes. The Supplement contains the closed form calculations, the proofs and the additional details of the numerical experiments.

\section{Methodology}\label{sec: methodology}

Let $\mc{S}=\{\s_1,\ldots,\s_\K\}$ be $K$ equidistant locations on $[0, 1]$, satisfying $\s_k=(k-1)/(\K-1)$. For individual $i=1,\ldots,n$, we observe the covariate vector $\bx_i\in\bbR^p$, scaled so that $\max_i\|\bx_i\|_2=1$, and the response $y_{i\s}\in[0,1]$ at each location $\s\in\mc{S}$. 
Given $\bx_i$, we assume that the responses at varying locations are independent, with
\begin{equation}\label{eqn: conditional density model}
  y_{i\s}\mid\bx_i,\theta \overset{ind}{\sim} f_\theta(\cdot\mid\bx_i,\s).
\end{equation}

\subsection{Model specification}\label{ssec: framework}

We construct the ExFR-LGP model by specifying the conditional log density as
\begin{equation}\label{eqn: model}
  \log f_\theta(y\mid\bx,\s) = g_1(y,w_{\bx,\s}) + g_2(y,\s) - \lpart_\theta(\bx,\s),
\end{equation}
where $\lpart_\theta(\bx,\s) =\log\int_0^1\exp\{g_1(t,w_{\bx,\s})+g_2(t,\s)\}\,dt$ is the normalizing constant. Here, $w_{\bx,\s}$ denotes the rescaled index
\begin{equation}
  w_{\bx,\s} = c((\bx^\T \beta(\s) + 1)/2), \qquad c(z) = \min(\max(z, 0), 1).
\end{equation}
The function $c$ clamps the range of $w_{\bx,\s}$ to $[0, 1]$ so that $g_1$ is defined on $[0, 1]^2$. The component $g_1$ captures the effect of $\bx$ through the location varying index $\bx^\T\beta(\s)$, and $g_2$ captures the density variation that does not depend on the covariates. 

For $g_1$ and $g_2$, we adopt the following basis representation:
\begin{align}\label{eqn: basis expansion g}
    g_l(z_1,z_2) = \sum_{i,j=0}^{\N_l}\xi^{(l)}_{ij}\,\phi^{(l)}_i(z_1)\,\phi^{(l)}_j(z_2), \qquad \xi^{(l)}=\{\xi^{(l)}_{ij}\}\in\bbR^{(\N_l+1)^2} \qquad (l=1,2).
\end{align}
The hat basis function $\phi^{(l)}_i(z)$ is defined as
\begin{equation}\label{eqn: hat basis}
  \phi^{(l)}_i(z) = \max(1 - |\N_lz - i|, 0).
\end{equation}
The resulting $g_l$ is a 2D piecewise linear function with nodal values satisfying $g_l(i/\N_l, j/\N_l)=\xi^{(l)}_{ij}$. Therefore, as suggested in \cite{maatouk2017gaussian}, taking $\xi^{(l)}_{ij}$ as the nodal values of a continuous GP defined on $[0, 1]^2$ leads to $g_l$ being a linear interpolation of a parent GP. This representation has been widely used for purposes such as computational scalability and shape constraints \citep{maatouk2017gaussian, griebel2010finite, ray2020efficient,kim2025adaptive,kim2026isotonic}. We use a parent GP with mean zero and covariance kernel $\sigma_l^2K^{(l)}$, with $K^{(l)}$ a correlation kernel, and this leads to $\xi^{(l)}\sim\cN(0,\sigma_l^2\Gamma^{(l)})$ with $\Gamma^{(l)}$ the correlation matrix of $K^{(l)}$ evaluated at knots, and we assign $\sigma_l^2\sim\mathrm{IG}(a_\sigma,b_\sigma)$.  The resulting $g_1$ and $g_2$ are finite-rank approximations of the parent GPs. Together with the log-density formulation in \eqref{eqn: model}, these representations define the ExFR-LGP model.

As an input of $g_1$, the rescaled index $w_{\bx, \s}$ is obtained from index coefficients over the locations in $\mc{S}$. For $\beta=\{\beta^\T(\s_1),\ldots,\beta^\T(\s_\K)\}^\T\in\bbR^{p\K}$, we assign a Gaussian prior $\beta\sim\cN(0,\Sigma)$ with 
\begin{equation}\label{eqn: covariance of beta}
    \Sigma = \sigma^2_{\beta}\, \Gamma_\beta \otimes \mathrm{I}_{p}, \qquad \Gamma_{\beta, kk'} = K_\beta(\s_k,\s_{k'}),
\end{equation}
where $K_\beta$ is a stationary positive definite kernel on $[0,1]$. Namely, we assign independent and identical GP prior over the locations for each covariate.

Under such structure, the exponent of $f_\theta$ is piecewise linear in $y$. Hence, its integral, the conditional mean and the quantiles are available in closed form (see Supplement \ref{sec: closed form} for details). Unlike the covariate dependent mixtures in Section~\ref{sec: introduction}, the smoothness of the log density in the response, the index and the location is stated directly as the GP priors on $g_1$ and $g_2$, and each effect enters through its own component. Multimodal functions can be modeled as well, thanks to the piecewise linear structure in \eqref{eqn: basis expansion g}. It is worth noting that under this finite-rank GP prior on $g_l$ induced from piecewise linear interpolation, the likelihood and relevant quantities are calculated in a closed form, without any further numerical approximation.

In this model, the parameter $\theta$ itself is not identifiable. For instance, as $g_1$ and $g_2$ share the basis functions, replacing $\big(\xi^{(1)}_{ij},\xi^{(2)}_{ij}\big)$ by $\big(\xi^{(1)}_{ij}+a_i,\xi^{(2)}_{ij}-a_i\big)$ for any $a_i$ leaves $f_\theta$ unchanged.
 Also, since $c(1-z)=1-c(z)$ and $\phi^{(1)}_j(1-z)=\phi^{(1)}_{\N_1-j}(z)$, replacing $\big(\beta,\xi^{(1)}_{ij}\big)$ by $\big(-\beta,\xi^{(1)}_{i,\N_1-j}\big)$ leaves $f_\theta$ unchanged.
Therefore, we summarize the effect of the covariates through $f_\theta$ (see Section~\ref{sec: real data}).

\subsection{Related and alternative models}\label{ssec: alternative models}

Several existing models share parts of this construction. The index is a one-dimensional version of the subspace projection LGP (spLGP) of \cite{tokdar2010bayesian}, which models $f(y\mid\bx)$ as a function of $(B^\T\bx,y)$ for a $p\times d$ matrix $B$ and selects $d$ from the data. Here, $B$ is common to all $\bx$ and the model has no location component. There is a literature utilizing finite representations of the GP used in LGP, such as orthogonal series \citep{verdinelli1998bayesian}. However, its integral does not enjoy a closed form. 

The representation \eqref{eqn: basis expansion g} is closest to that of \cite{griebel2010finite}, who compute the posterior mode of an LGP density by a Galerkin approximation in piecewise multilinear functions on a uniform grid, evaluate its integrals by Gauss quadrature and bound the discretization error. \ourmodel{} uses the same elements in two dimensions to define the prior, such that the likelihood is exact, enabling sampling of the whole posterior. 

\cite{gautier2025continuous} and \cite{chen2026logistic} allow the density to vary with location by taking the location as an input of the GP, and fit it by numerical integration and by a score matching loss, respectively. In Sections~\ref{sec: simulation} and \ref{sec: real data}, we also compare \ourmodel{} with a full grid LGP, which keeps the finite-rank construction of \ourmodel{} but replaces $g_1(y,w_{\bx,\s})+g_2(y,\s)$ in \eqref{eqn: model} by a single function $g(y,\bx,\s)$ and represents it by linear interpolation on a tensor grid in $(y,\bx,\s)$, which is the discretization of \cite{griebel2010finite} applied to $(y,\bx,\s)$. We also compare it with the spLGP of \cite{tokdar2010bayesian} and the logit stick-breaking mixture of \cite{rigon2021tractable}.

\section{Theoretical Properties}\label{sec: posterior contraction}

Under the model \eqref{eqn: model} with the priors in Section~\ref{sec: methodology}, we derive the posterior contraction rate, which quantifies how fast the posterior concentrates around the true density. We build on the general results of \cite{ghosal2000convergence,ghosal2007convergence,ghosal2007posterior,van2008rates}.

We start with some notation. For $\alpha>0$, let $\alpha_0$ be the largest integer strictly smaller than $\alpha$, and for $k=(k_1,k_2)$ let $D^k=\partial^{|k|}/\partial z_1^{k_1}\partial z_2^{k_2}$ with $|k|=k_1+k_2$. We denote by $\bbC^\alpha(\mc{Z})$ the $\alpha$-H\"older space on a closed rectangle $\mc{Z}\subset\bbR^2$, the set of functions $g:\mc{Z}\to\bbR$ whose partial derivatives $D^k g$ with $|k|\le\alpha_0$ exist and are bounded, and whose derivatives of order $\alpha_0$ satisfy $|D^k g(z)-D^k g(z')|\le L\|z-z'\|_2^{\alpha-\alpha_0}$ for some $L>0$ and all $z,z'\in\mc{Z}$, and write $\bbC^\alpha[0,1]^2$ for $\mc{Z}=[0,1]^2$.
Define the parameter space $\Theta = \set{\theta = (\xi^{(1)},\xi^{(2)},\beta) : \xi^{(1)},\xi^{(2)}\in\bbR^{(\N+1)^2}, \beta\in \bbR^{p\K}}$ and 
let $\mc{F} = \set{f_\theta: \theta\in \Theta}$ be the collection of conditional densities $f_\theta(\cdot\mid\bx,\s)$ $(\|\bx\|_2 \le 1, \s \in \mc{S})$ in \eqref{eqn: model}. A prior on $\Theta$ defined in Section \ref{sec: methodology} induces a prior on $\mc{F}$, which we refer to as $\Pi$. We write $\|g\|_{\infty} = \sup_{(z_1,z_2)\in[0,1]^2}\abs{g(z_1,z_2)}$ for $g:[0,1]^2\to \bbR$ and $\|v\|_{\infty} = \max_j\abs{v_j}$ for a vector $v$. Finally, let $f_0(y \mid \bx, \s)$ denote the true conditional density, and denote the data $\cD_n=\set{(\bx_i, y_{i\s}): 1\leq i\leq n, \s\in \mc{S}}$. For $z,z'\in[0,1]^d$ and $r=\|z-z'\|_2$, the RBF kernel with inverse length-scale $\kappa>0$ is $K^{\mrm{RBF}}_\kappa(z,z')=\exp(-\kappa^2r^2)$, and the Mat\'{e}rn kernel with smoothness $\nu>0$ and length-scale $\ell>0$ is $K^{\mrm{M}}_{\nu,\ell}(z,z')=\{2^{1-\nu}/\Gamma(\nu)\}(\sqrt{2\nu}\,r/\ell)^\nu\,\mrm{B}_\nu(\sqrt{2\nu}\,r/\ell)$, where $\Gamma(\cdot)$ is the gamma function and $\mrm{B}_\nu$ is the modified Bessel function of the second kind.

We first define the Hellinger distance between two conditional densities $f_1$ and $f_2$ as:
\begin{align}
    d^2_{h}(f_1,f_2) &= \frac{1}{\K}\sum_{\s\in\mc{S}}\int\p{\sqrt{f_1(y\m \bx, \s)} - \sqrt{f_2(y\m \bx, \s)}}^2 dy \,d \mc{H}_0(\bx),\label{eqn: hellinger}
\end{align}
with $\mc{H}_0$ denoting true distribution of $\bx$. This is the Hellinger distance between the joint distributions of $(y,\bx,\s)$ when $\bx\sim\mc{H}_0$ and $\s$ is uniform on $\mc{S}$. Without locations, it is the distance used for conditional densities \citep{shen2016adaptive, norets2017adaptive}.

We then consider the posterior contraction rate with respect to Hellinger distance $d_h$. If the following holds for some constant $C>0$:
\begin{align}\label{eqn: posterior contraction}
    \lim_{n \to \infty}\bbE_{f_0}\Pi(f_{\theta}\in\mc{F}:\,d_h(f_\theta,f_0)\geq C\varepsilon_n\mid \cD_n) = 0,
\end{align}
$\varepsilon_n$ is called a posterior contraction rate. This implies that as $n$ increases, the posterior distribution of $f$ is concentrated around a radius $C \varepsilon_n$ ball with respect to $d_h$.

Posterior contraction rates are well developed for Bayesian density estimation \citep{ghosal2000convergence, shen2001rates, shen2013adaptive}. For conditional densities, most results give consistency \citep{tokdar2010bayesian, pati2013posterior, norets2014posterior}. Rates are given by \cite{shen2016adaptive} and \cite{norets2017adaptive}, whose exponents depend on the number of relevant covariates, and by \cite{jiang2007bayesian} for generalized linear models. Rates under a single index are available for regression \citep{dhara2020new, roy2020high}, and \cite{tokdar2011dimension} gives the rate $n^{-\alpha/(2\alpha+d_0+1)}$ for spLGP under a rank $d_0$ projection with $\alpha>1$, for a squared exponential GP rather than the node approximation used in computation. \cite{odin2026contraction} let the dimensions grow with $n$ for regression and density estimation.

In this section, we assume that the observations $\{\bx_i, (y_{i\s})_{\s\in\mc{S}}\}$, $i=1,\ldots,n$ are obtained as follows. First, $\{\bx_i \}$ are independently and identically sampled from the distribution $\mc{H}_0$ supported on $\{\bx:\|\bx\|_2\leq 1\}$. Given each $\bx_i$, $(y_{i\s})_{\s\in\mc{S}}$ are independently drawn from the density $f_0(\cdot\mid\bx_i,\s)$. The numbers $p$ and $\K$ are fixed. Assume that the true density function is represented as $f_{0}(y\mid \bx,\,\s) = \Upsilon_0(y,\bx,\s)/\int_\cY \Upsilon_0(y,\bx,\s)dy$ with some $\Upsilon_0$. Then the following holds.

\begin{theorem}\label{thm:post_cont}
Suppose that $\log\Upsilon_0(y, \bx, \s) = g_1^{(0)}(y, \bx^\T\beta^{(0)}(\s)) + g_2^{(0)} (y, \s)$, where $g^{(0)}_1$ and $g^{(0)}_2$ are $\alpha$-H\"older for some $\alpha>0$. Assume that $K^{(l)}=K^{\mrm{RBF}}_{\kappa_l}$ with $\sigma_l^2$ fixed $(l=1,2)$ and $\N_1=\N_2=\N\asymp n$, and that $\kappa_1,\kappa_2$ are i.i.d.\ with a density $\pi_\kappa$ on $[1,\infty)$ satisfying
\begin{equation}
  C_1\kappa^{p_\kappa}\exp(-D_1\kappa^2\log^4\kappa)\le\pi_\kappa(\kappa)\le C_2\kappa^{p_\kappa}\exp(-D_2\kappa^2\log^4\kappa)
\end{equation}
for large $\kappa$ with constants $C_1,C_2,D_1,D_2>0$ and $p_\kappa\ge0$. Then \eqref{eqn: posterior contraction} holds for \ourmodel{} in \eqref{eqn: model} with
\begin{align}\label{eqn: rate}
\varepsilon_n = n^{-\frac{\alpha}{2\alpha + 2}} (\log n)^{\frac{2\alpha}{\alpha+1}}.
\end{align}
\end{theorem}
\begin{proof}
 Refer to Supplement \ref{sec: proof of theorem} for the proof of the theorem.
\end{proof}

Since $\N$ and $\pi_\kappa$ do not depend on the unknown smoothness $\alpha$, the posterior contraction rate in Theorem \ref{thm:post_cont} is an adaptive result. The rate $n^{-\alpha / (2\alpha + 2)}$ given in Theorem \ref{thm:post_cont} is the minimax rate for an $\alpha$-smooth bivariate function, up to a logarithmic factor \citep{tsybakov2004introduction}. The proof reduces the problem to the bivariate functions $g_1$ and $g_2$ and applies the results for GPs with RBF kernels \citep{van2008rates,van2009adaptive}. Corollary~\ref{cor:growing_p} in Supplement~\ref{ssec: growing p} shows that the same rate holds when $p$ grows with $n$, as long as $p\K\log n=o(n\varepsilon_n^2)$ with $\varepsilon_n$ given in \eqref{eqn: rate}, when the covariates and the true index are bounded.

\section{Posterior computation}\label{sec: posterior computation}

Next, we present a posterior computation framework for the model \eqref{eqn: model}. For modeling flexibility, we first introduce an extension of \ourmodel{} that models potential deviations from \eqref{eqn: model}. We then outline a Gibbs sampler for both models, updating the parameters with a Gaussian prior by ESS.

\subsection{Extended Exact Likelihood Finite-Rank LGP}\label{ssec: generalized model}

To model $f$ that deviates from the form \eqref{eqn: model}, we consider the extended exact likelihood finite-rank LGP (\ourgmodel{}),
\begin{equation}\label{eqn: generalized model}
  \log f_\theta(y\mid\bx,\s) = g_1(y,w_{\bx,\s}) + g_2(y,\s) + \delta(y,\bx,\s) - \lpart_\theta(\bx,\s),
\end{equation}
where $\delta$ models this deviation and $\lpart_\theta$ is the log-normalizing constant of the new exponent. Since $\delta$ depends on $\bx$ directly rather than through the index, it can capture departures from both the single index and the additive form of \eqref{eqn: model}. In our implementation, $\delta(y,\bx,\s)=\sum_{j=0}^{\N_\delta}\phi^{\delta}_j(y)\,h_j(\bx,\s)$, where $\phi^{\delta}_j$ is the hat basis \eqref{eqn: hat basis} with $\N_\delta$ in place of $\N_l$. The functions $h_0,\ldots,h_{\N_\delta}$ are assigned a zero mean GP prior with $\mathrm{cov}\{h_j(\bx,\s),h_{j'}(\bx',\s')\}=\sigma_\delta^2K_y(j/\N_\delta,j'/\N_\delta)\,K_{\bx}(\bx,\bx')\,K_{\s}(\s,\s')$, where $K_y$, $K_{\bx}$ and $K_{\s}$ are Mat\'{e}rn correlation kernels, and $\sigma_\delta^2\sim\mathrm{IG}(a_\delta,b_\delta)$. We set $(a_\delta,b_\delta)$ to keep $\sigma_\delta^2$ and the discrepancy $\delta$ small, unless it is strongly suggested by the data. The same priors from Section~\ref{ssec: framework} are adopted on $g_1$, $g_2$ and $\beta$. With $\delta\equiv0$, this reduces to \ourmodel{}. The exponent in \eqref{eqn: generalized model} remains piecewise linear in $y$, so the closed forms of the quantities such as integration, quantiles and moments given in Supplement \ref{sec: closed form} can be calculated in a similar manner.

We then present a posterior sampler for \ourgmodel{}. Such a sampler can naturally be used for \ourmodel{} by simply ignoring the sampling steps involving $\delta$.

\subsection{Gibbs sampler}\label{ssec: sampler}

We use the same grid for $g_1$ and $g_2$, $\N_1=\N_2=\N$, and write $\phi_i=\phi^{(l)}_i$. For the GP kernels $K^{(l)}$ of $g_l$, from which the priors on $\xi^{(1)}$ and $\xi^{(2)}$ are constructed, we use Mat\'{e}rn kernels $K^{\mrm{M}}_{\nu,\ell}$ defined in Section~\ref{sec: posterior contraction}, since the covariance matrix of the RBF kernel in Theorem~\ref{thm:post_cont} is nearly singular on a fine grid \citep{schaback1995error,wendland2004scattered}. We also fix $\N$ for computational convenience.
Here, we fix the smoothness and the length-scale of the Mat\'{e}rn kernels at the values given in Section~\ref{sec: simulation}. Fixing the smoothness is common practice, with standard choices $\nu=3/2$ and $5/2$. Larger values are hard to distinguish from one another with finite noisy data \citep{rasmussen2006gaussian}. The length-scale is set so that the prior correlation between the two farthest grid points of $g_l$ is about $0.001$, which gives a weakly informative prior. Under a Mat\'{e}rn kernel, the variance $\sigma^2$ and the length-scale cannot be consistently estimated separately, and predictions depend mainly on their ratio $\sigma^2/\ell^{2\nu}$ \citep{zhang2004inconsistent}, so we fix $\ell$ and learn $\sigma_l^2$. Since only $\beta(\s)/\|\beta(\s)\|_2$ enters the likelihood as described below, $\sigma^2_\beta$ has no effect and we set it to one.

Therefore, the model parameters to sample are $\theta=(\xi^{(1)},\xi^{(2)},\beta,\sigma_1^2,\sigma_2^2)$, and the Gibbs sampler alternates between $\beta$, $(\xi^{(1)},\xi^{(2)})$ and $(\sigma_1^2,\sigma_2^2)$. For each block, $\theta_{-}$ denotes the other parameters at their current values. Under \eqref{eqn: model}, the log-likelihood is
{\small
\begin{align}\label{eqn: log likelihood}
        \log\cL(\theta)
    &= \sum_{\s\in\mc{S}} \sum^n_{i=1} \Bigg[\bigg\{\sum^\N_{j,k=0} \xi^{(1)}_{jk}\,\phi_j(y_{i\s})\phi_k(w_{\bx_i,\s})\bigg\}+ \bigg\{\sum^\N_{j,k=0} \xi^{(2)}_{jk}\,\phi_j(y_{i\s})\,\phi_k(\s)\bigg\} - \lpart_{\theta}\p{{\bx_i,\s}}\Bigg].
\end{align}
}
Since $\xi^{(1)}$ and $\xi^{(2)}$ are independent with Gaussian priors, their conditional posterior is
\begin{equation}
  \pi(\xi^{(1)},\xi^{(2)} \mid \theta_{-}, \cD_n) \propto \cN(\xi^{(1)};0,\sigma_1^2\Gamma^{(1)})\,\cN(\xi^{(2)};0,\sigma_2^2\Gamma^{(2)})\,\cL(\theta),
\end{equation}
a Gaussian prior times a likelihood, so we update $(\xi^{(1)},\xi^{(2)})$ jointly by one ESS step \citep{murray2010elliptical}. While separate update of $\xi^{(1)}$ and $\xi^{(2)}$ in a Gibbs sampler is also valid, we sample them jointly as $g_1$ and $g_2$ are correlated a posteriori since they share the response $y$. Likewise, $\pi(\beta\mid\theta_{-},\cD_n)\propto\cN(\beta;0,\Sigma)\,\cL(\theta)$, and $\beta$ is also updated by ESS. The inverse gamma prior of $\sigma_l^2$ is conjugate,
\begin{equation}\label{eqn: variance conditional}
  \sigma_l^2 \mid \theta_{-}, \cD_n \sim \mathrm{IG}\big\{a_\sigma+(\N+1)^2/2,\; b_\sigma+\xi^{(l)\T}(\Gamma^{(l)})^{-1}\xi^{(l)}/2\big\} \qquad (l=1,2),
\end{equation}
so $\sigma_l^2$ is drawn after each ESS step of $(\xi^{(1)},\xi^{(2)})$.

The scales of $\beta(\s)$ and $g_1$ are not separately identifiable, since a rescaling of $\beta(\s)$ can be offset by rescaling the second argument of $g_1$. To improve identifiability, we apply to the sampled $\beta(\s)$ an on-the-fly normalization to unit $L_2$ norm, widely used in spatial statistics \citep{rue2005gaussian, mak2016regional}. Namely, we obtain the posterior sample of $\beta$ from ESS updates as an unconstrained vector, and use $\beta(\s)/\|\beta(\s)\|_2$ in the likelihood and in all summaries. Since $\|\bx\|_2\leq 1$, the index $w_{\bx,\s}$ then always lies in $[0,1]$ without clamping, so the grid of $g_1$ covers every value of the index.

Since $\nu$ and $\ell$ are fixed, the Cholesky factors of $\Gamma^{(1)}$, $\Gamma^{(2)}$ and $\Sigma$ are computed once and only rescaled by $\sigma_l$, after which each prior draw of $\xi^{(1)}$ or $\xi^{(2)}$ and each draw from \eqref{eqn: variance conditional} costs $\bigO({M}^2)$ with ${M}=(\N+1)^2$. Each ESS step also evaluates \eqref{eqn: log likelihood}, usually several times. For each $(i,\s)$, at most two of the $\phi_k(w_{\bx_i,\s})$ are nonzero; so, with $\lpart_\theta$ in closed form (Supplement \ref{ssec: normalizing constant}), one evaluation costs $\bigO\{n\K(\N+p)\}$. By the Kronecker structure of $\Sigma$, a prior draw of $\beta$ costs $\bigO(p\K^2)$. Since the grids of $g_1$ and $g_2$ do not depend on $p$, the covariates enter only through the index and $\beta$, such that the cost per iteration grows linearly in $p$ rather than exponentially as for a grid over the covariates.

For \ourgmodel{}, the likelihood depends on $h_j$ only at the observed pairs $(\bx_i,\s)$. So, we sample the latent values $H=\{h_j(\bx_i,\s)\}\in\bbR^{(\N_\delta+1)\times n\times\K}$, with $\mathrm{vec}(H)\sim\cN(0,\sigma_\delta^2\,\Gamma_{\s}\otimes\Gamma_{\bx}\otimes\Gamma_{y})$, where $\Gamma_{y}$, $\Gamma_{\bx}$ and $\Gamma_{\s}$ are the correlation matrices of $K_y$, $K_{\bx}$ and $K_{\s}$ over the knots $j/\N_\delta$, the covariate vectors $\bx_1,\ldots,\bx_n$ and the locations in $\mc{S}$. Each term of \eqref{eqn: log likelihood} gains $\sum_{j=0}^{\N_\delta}h_j(\bx_i,\s)\,\phi^{\delta}_j(y_{i\s})$, and $\lpart_\theta$ is computed from the exponent in \eqref{eqn: generalized model}. Since $H$ has a Gaussian prior, it is updated by one ESS step after $(\xi^{(1)},\xi^{(2)})$, and the inverse gamma prior of $\sigma_\delta^2$ is conjugate,
\begin{equation}\label{eqn: delta variance conditional}
  \sigma_\delta^2 \mid \theta_{-}, \cD_n \sim \mathrm{IG}\big\{a_\delta+(\N_\delta+1)n\K/2,\; b_\delta+\mathrm{vec}(H)^\T(\Gamma_{\s}\otimes\Gamma_{\bx}\otimes\Gamma_{y})^{-1}\mathrm{vec}(H)/2\big\}.
\end{equation}
By the Kronecker structure, only the Cholesky factors of $\Gamma_{y}$, $\Gamma_{\bx}$ and $\Gamma_{\s}$ are needed, and each prior draw of $H$ and the quadratic form in \eqref{eqn: delta variance conditional} cost $\bigO\{(\N_\delta+1)n\K(\N_\delta+n+\K)\}$. At a new covariate vector $\bx$, $h_j(\bx,\s)$ is drawn from its Gaussian conditional distribution given $H$. Algorithm \ref{alg: sampler} outlines the Gibbs sampler, where \texttt{ex2.flg} indicates whether \ourgmodel{} is used instead of \ourmodel{}.

\begin{algorithm}[!t]
\caption{Sampling from the posterior of $\theta$}\label{alg: sampler}
\SetAlgoLined
\SetArgSty{textup}
\textbf{Input:} data $\cD_n$, number of samples $T_{\iter}$, initial value $\theta_0$ and model flag \texttt{ex2.flg}, with $H=0$ in $\theta_0$\\
 \For{$t=1,\ldots,T_{\iter}$}{
 Update $\beta_t$ from $\pi(\beta\mid\theta_{-},\cD_n)$ by ESS, using $\beta_t(\s)/\|\beta_t(\s)\|_2$ in $w_{\bx,\s}$\;
 Update $(\xi^{(1)}_t,\xi^{(2)}_t)$ from $\pi(\xi^{(1)},\xi^{(2)}\mid\theta_{-},\cD_n)$ by one ESS step\;
 Draw $\sigma^2_{l,t}$ from $\pi(\sigma_l^2\mid\theta_{-},\cD_n)$ in \eqref{eqn: variance conditional} $(l=1,2)$\;
 \If{\texttt{ex2.flg} $=$ \texttt{true}}{
 Update $H_t$ from $\pi(H\mid\theta_{-},\cD_n)$ by one ESS step\;
 Draw $\sigma^2_{\delta,t}$ from $\pi(\sigma_\delta^2\mid\theta_{-},\cD_n)$ in \eqref{eqn: delta variance conditional}\;
 }
 }
\textbf{Result:} $\theta_t=(\xi^{(1)}_t,\xi^{(2)}_t,\beta_t,\sigma^2_{1,t},\sigma^2_{2,t})$, with $(H_t,\sigma^2_{\delta,t})$ added for \ourgmodel{}, $t=1,\ldots,T_{\iter}$.
\end{algorithm}

\section{Simulation studies}\label{sec: simulation}

We now investigate the finite-sample performance of \ourmodel{} in numerical experiments, using synthetic data. We fix $p=5$ and take $\beta_1(\s) = \s^2+1,\,\beta_2(\s) = (1-\s)^2,\,\beta_3(\s) = 4\s(1-\s),\beta_4(\s) =  2(\s-0.75)^2-1,\, \beta_5(\s) =  1.5 + 4\sin(\s-0.5)$, where $\s$ denotes the normalized location on $[0,1]$, and rescale $\beta(\s)$ to unit norm at each location. Each entry of the design matrix $\bX$ is generated from $\Unif(-5,5)$, centered columnwise, and normalized by $\max_i \norm{\bx_i}$. Since $\|\beta(\s)\|_2=1$ and $\|\bx\|_2\le1$, $w_{\bx,\s}=\{\bx^\T\beta(\s)+1\}/2$ lies in $[0,1]$. Given $(\bx,\s)$, the response is generated from one of the following two densities.
\begin{itemize}
    \item Case 1 (unimodal)
    \begin{equation}\label{eqn: unimodal}
        (y\mid \bx, \s) \sim \mrm{Beta}[\exp(2w_{\bx,\s}+1),(18-8\s)].
    \end{equation}
    \item Case 2 (bimodal)
    \begin{equation}\label{eqn: bimodal}
        (y\mid \bx, \s) \sim
        \begin{cases}
        \mrm{Beta}\big[10\exp\set{0.5\,(k/\K-1/2)},4\,k^{3/4}\big]&\text{with probability $1/[1+\exp\set{\bx^\T\beta(\s)}]$}\\
        \mrm{Beta}\big[32\sin(w_{\bx,\s}),8(k/\K+1)^{1/2}\big] &\text{with probability $1/[1+\exp\set{-\bx^\T\beta(\s)}]$}
        \end{cases}
    \end{equation}
    where $k$ is the index of the location $\s=\s_k$.
\end{itemize}

The log density in Case 1 is a function of $(y,w_{\bx,\s})$ plus a function of $(y,\s)$, so it has the form of \eqref{eqn: model}. Case 2 is a mixture of two Beta densities whose weights depend on $\bx^\T\beta(\s)$, so the number and the positions of the modes change with $(\bx,\s)$. It does not have the form of \eqref{eqn: model}, and a single Beta density cannot capture it.

For each $n\in\{50,125,250\}$ and $\K\in\{5,10,20,30\}$, we generate $50$ datasets. We randomly select $80\%$ of the individuals to train the model and keep the remaining $n_2$ individuals, $\bx_{i,\mrm{test}}$ $(i=1,\ldots,n_2)$, to assess the fit. Each $f_{\theta_t}$ integrates to one since its normalizing constant is exact (Supplement \ref{ssec: normalizing constant}). We take the pointwise median of $f_{\theta_t}(\cdot\m\bx,\s)$ over the retained states $\theta_t$, normalized to integrate to one, as $\widehat{f}(\cdot\m\bx,\s)$, and the pointwise $(2.5\%, 97.5\%)$ quantiles give a 95\% credible band for $f(\cdot\m\bx,\s)$.

The error of $\widehat{f}$ is the squared Hellinger distance \eqref{eqn: hellinger} to the true density $f_0$, with $\mc{H}_0$ replaced by the empirical distribution of the $\bx_{i,\mrm{test}}$ and the integral over $y$ computed on an equispaced grid of $10^4+1$ points. For illustration, we also generate one dataset with $n=200$ and $\K=20$ from each case. 

We compare \ourmodel{} with three methods: the full grid LGP of Section~\ref{ssec: alternative models}, the spLGP of \cite{tokdar2010bayesian} and the logit stick-breaking prior (LSBP) mixture of \cite{rigon2021tractable}. For \ourmodel{}, we use $\N=50$ and Mat\'{e}rn kernels with smoothness $\nu=5/2$ and length-scale $\ell=0.292$ for $K^{(l)}$ and $K_\beta$, chosen following Section~\ref{sec: posterior computation}, with $\sigma^2_\beta=1$, $a_\sigma=2$ and $b_\sigma=1$. The full grid LGP uses the same kernel and variance prior, with a coarse grid of $3$ points per covariate and $11$ points for $\s$, since the number of grid values grows exponentially in $p$. The spLGP uses a single index treating the location as one more covariate, with the same grid, kernel and priors as \ourmodel{}. LSBP is fitted to the logit of $y$ as in \cite{rigon2021tractable}, with kernel means linear in $(\bx,\s)$ and mixing weights depending on $\bx$ and a spline basis of $\s$. All methods are run for $2000$ burn-in and $2000$ posterior states, and Supplement \ref{ssec: benchmarks} gives the implementation details.

In Figure \ref{fig:sim_comp}, we overlay the estimated densities of the four methods with the true density for different individuals at the same locations of the illustrative datasets. In Case 1, \ourmodel{} captures the truth, and its $95\%$ credible bands cover it.
 In Case 2, which is misspecified, \ourmodel{} still recovers both the unimodal and the bimodal shapes, while the full grid LGP shows a single broad mode in all three panels. In panel (f), the spLGP, LSBP and the full grid LGP depart from the truth, while the band of \ourmodel{} covers it.

\begin{figure}[t!]
    \centering
  \centering
\begin{subfigure}{0.330\textwidth}
  \centering
  \includegraphics[scale= 0.15]{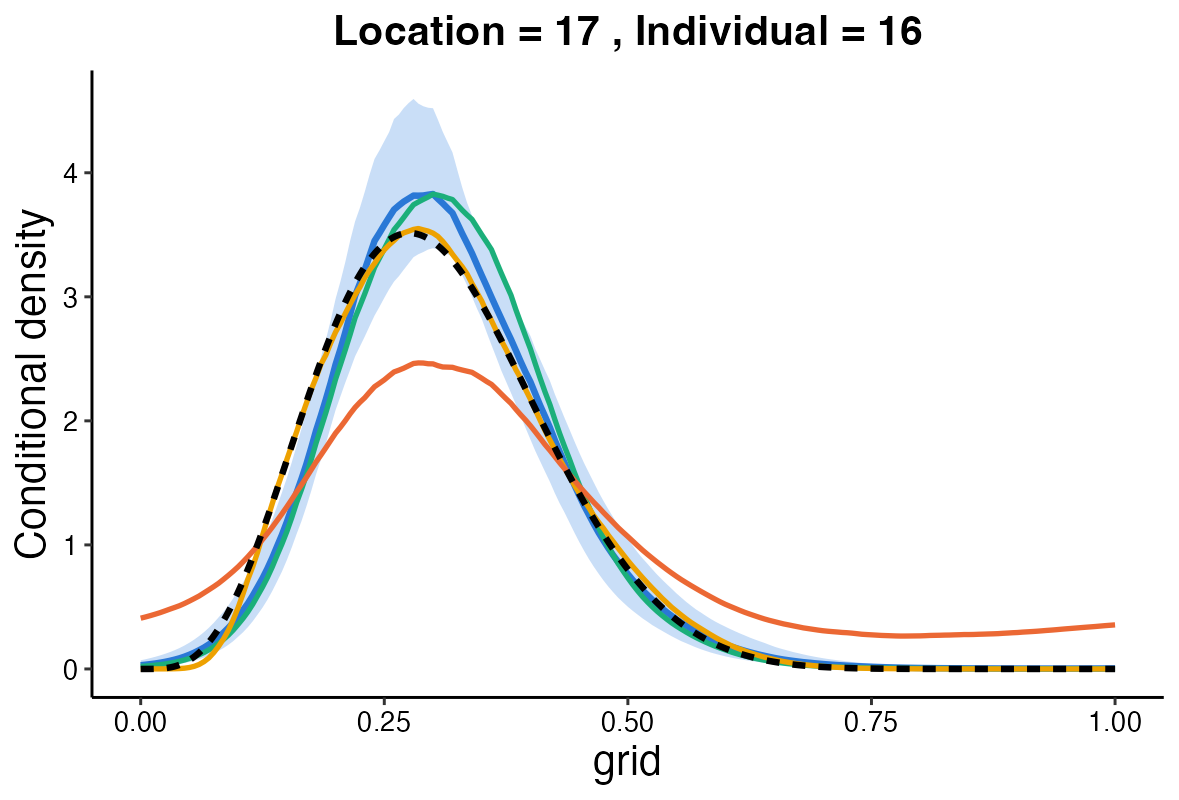}
  \caption{}
  \label{fig:a1}
\end{subfigure}%
  \centering 
\begin{subfigure}{0.330\textwidth}
  \centering
    \includegraphics[scale= 0.15]{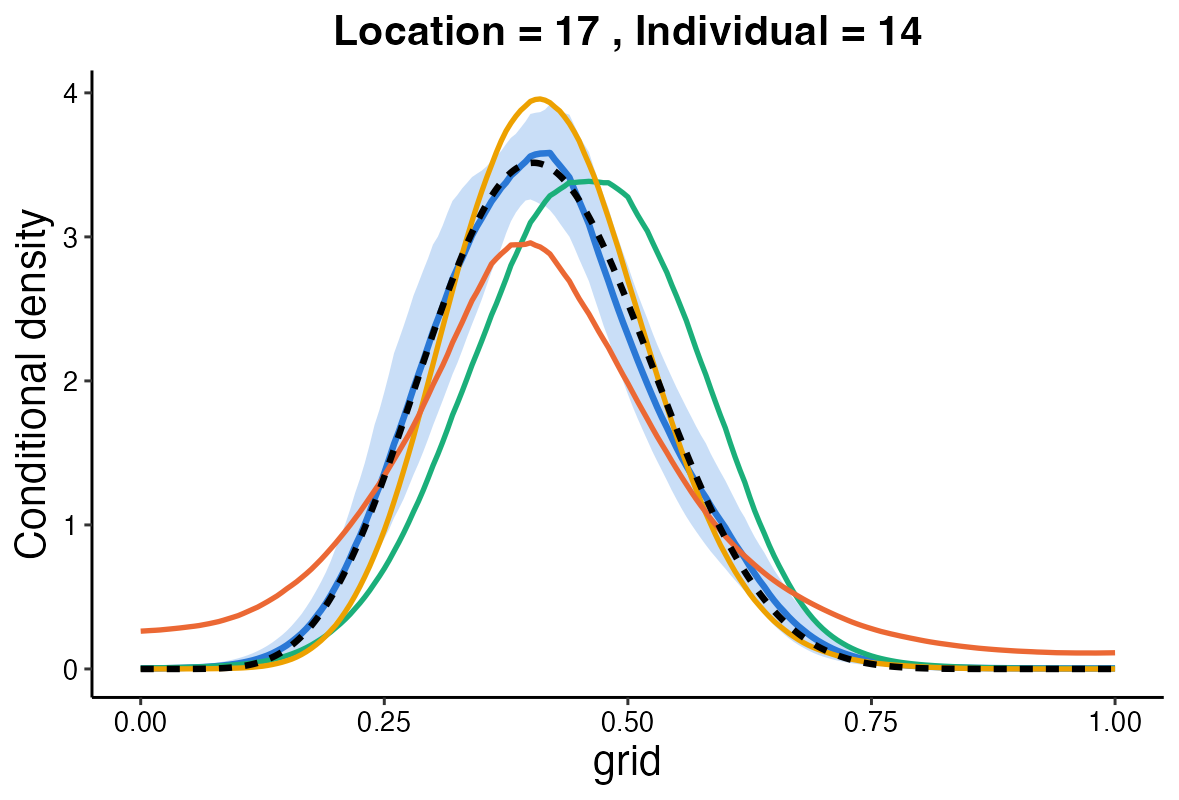}
  \caption{}
  \label{fig:b1}
\end{subfigure}%
\begin{subfigure}{0.330\textwidth}
  \centering
  \includegraphics[scale= 0.15]{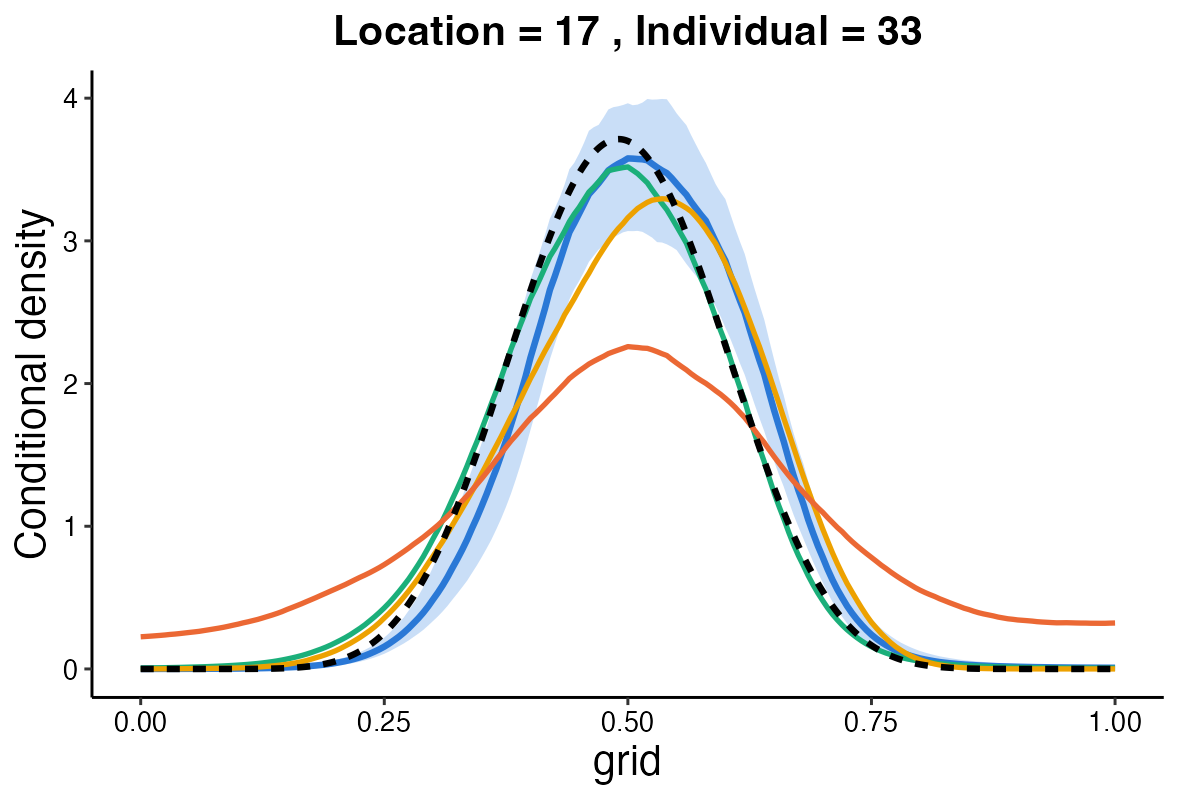}
  \caption{}
  \label{fig:c1}
\end{subfigure}%
\centering\\
\begin{subfigure}{0.330\textwidth}
  \centering
  \includegraphics[scale= 0.15]{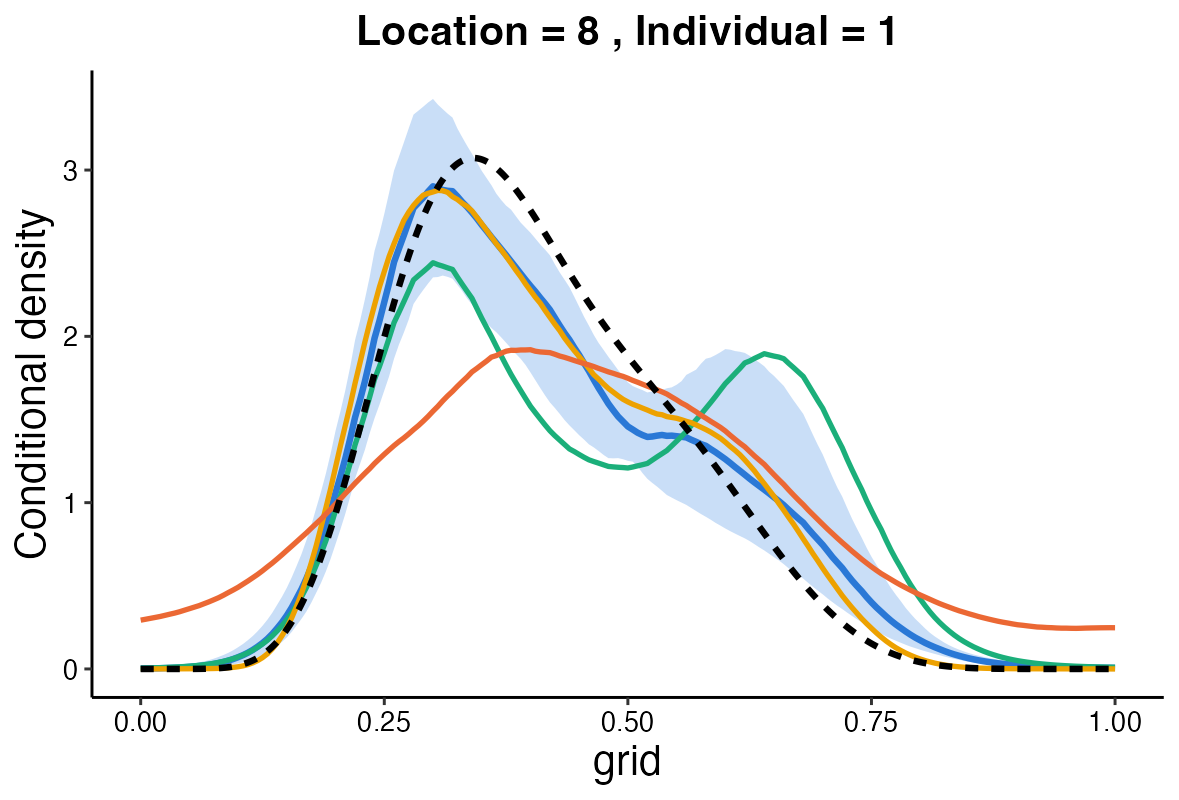}
  \caption{}
  \label{fig:d1}
\end{subfigure}%
  \centering 
\begin{subfigure}{0.330\textwidth}
  \centering
    \includegraphics[scale= 0.15]{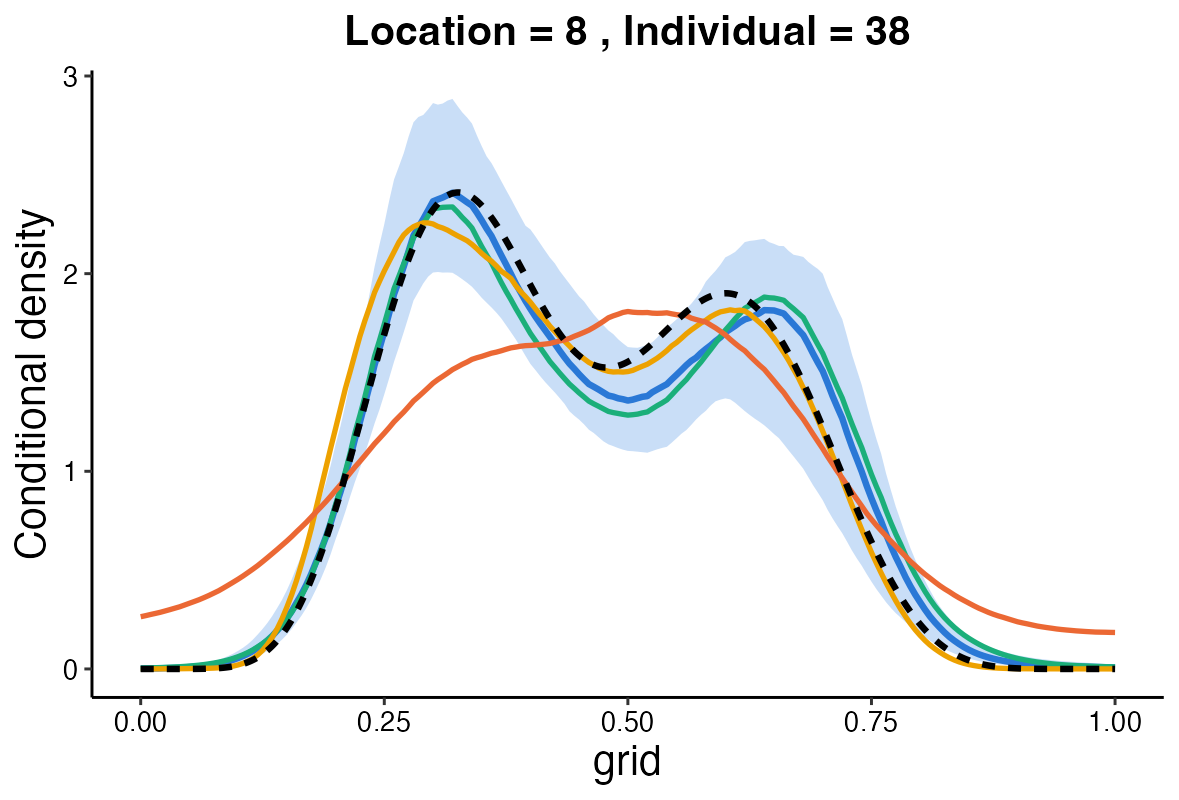}
  \caption{}
  \label{fig:e1}
\end{subfigure}%
\begin{subfigure}{0.330\textwidth}
  \centering
  \includegraphics[scale= 0.15]{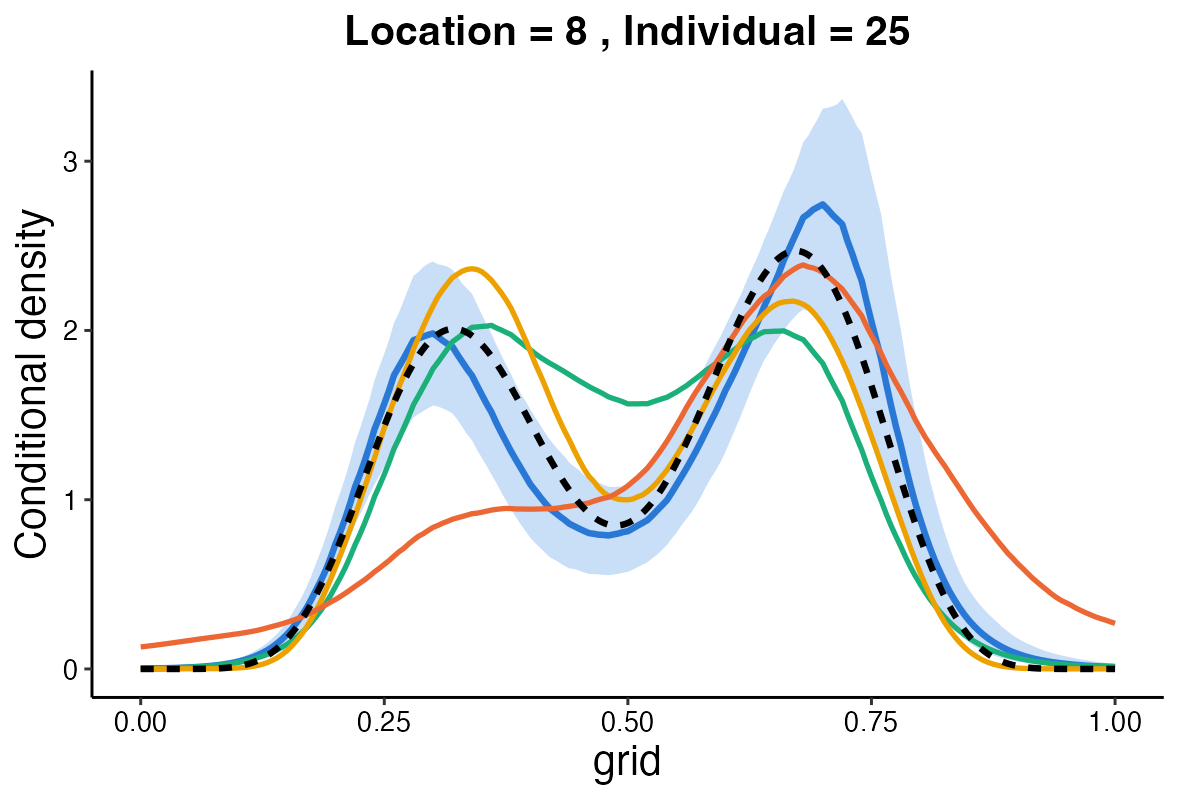}
  \caption{}
  \label{fig:f1}
\end{subfigure}%
\\
\includegraphics[width=0.8\textwidth]{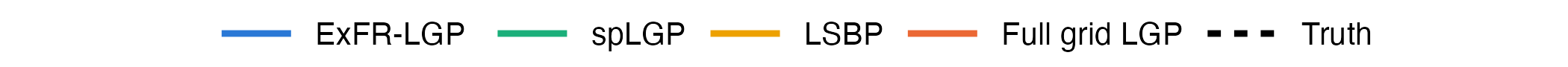}
\caption[Comparison of estimated densities]{The black dashed line represents the true density, and the blue line is the \ourmodel{} estimate with 95\% credible band in shaded blue for different individuals at the same location. Upper panels [(a), (b), (c)] and lower panels [(d), (e), (f)] correspond to densities in \eqref{eqn: unimodal} and \eqref{eqn: bimodal} respectively. The green, yellow and orange lines are the posterior medians of the spLGP, LSBP and the full grid LGP, respectively. The individuals are the test points at the 10th, 50th and 90th percentiles of $\bx^\T\beta(\s)$.}
\label{fig:sim_comp}
\end{figure}

Figure \ref{fig:discrepancy} shows the error over the $50$ replications against $n$ on a log scale. For fixed $\K$, the error of \ourmodel{} drops as $n$ increases in both cases, in line with Theorem \ref{thm:post_cont}, and in Case 1 it also drops as $\K$ increases, since a larger $\K$ gives more responses to estimate $g_1$ and $g_2$, which are shared by all locations. The full grid LGP has a larger error in all replications. The spLGP has a larger error in almost all replications, and its error decreases slowly in $n$, since its single index cannot let the effect of $\bx$ change with $\s$. LSBP has a larger error than \ourmodel{} in Case 1 and a similar error in Case 2, where the truth is itself a mixture. In both cases, the $95\%$ credible band of \ourmodel{} covers the truth more often than those of the other methods (Supplement \ref{ssec: coverage}). When the responses of one individual are correlated along the tract, the predictive intervals of \ourmodel{} keep their nominal coverage, while the error grows and the credible band becomes narrower than the actual uncertainty (Supplement \ref{ssec: dependence}).

\begin{figure}[t!]
\centering
\includegraphics[width=\textwidth]{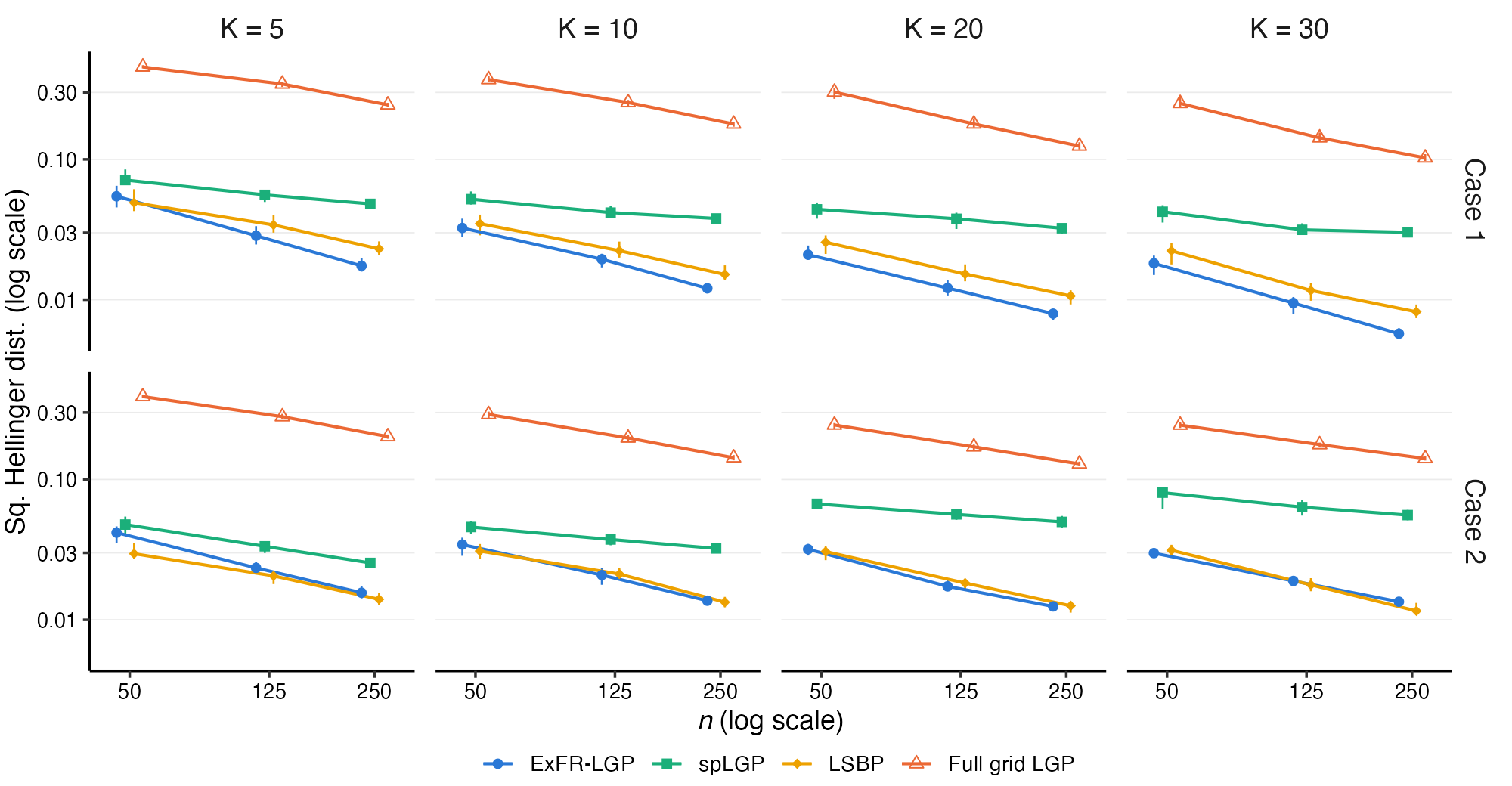}
\caption{Median (point) and interquartile range (bar) over $50$ replications of the squared Hellinger distance between the true and the estimated density, averaged over the test individuals and the $\K$ locations, against $n$, with both axes on a log scale. The upper row is Case 1 \eqref{eqn: unimodal} and the lower row is Case 2 \eqref{eqn: bimodal}, and the columns are the values of $\K$. Blue circles are \ourmodel{}, green squares the spLGP, yellow diamonds LSBP and orange triangles the full grid LGP.}
\label{fig:discrepancy}
\end{figure}

\section{Application: Diffusion Tensor Imaging}\label{sec: real data}
We applied \ourmodel{} in \eqref{eqn: model} to the ADNI data described in Section~\ref{sec: introduction}, obtained from the ADNI database (\url{https://adni.loni.usc.edu/}). The dataset comes from processed {\em diffusion-weighted imaging (DWI)} data \citep{bammer2003basic}. The FA maps were processed with the TBSS pipeline \citep{smith2006tract} as in \citet{luo2016single} and \citet{zhang2021high}, and each individual's FA was measured at $83$ grid points along the skeleton of the midsagittal corpus callosum. The FA measurement is a fraction between $0$ and $1$, where high FA values indicate healthy tissue and lower values typically indicate microstructural damage within coherent fiber bundles such as the corpus callosum \citep{basser1995inferring, sexton2011meta}.

We are interested in estimating the density of FA at a specific location given the covariates $\bx$. Here $\bx$ consists of seven covariates, four binary and three continuous. The binary covariates are sex (1 if male), handedness (1 if left-handed), AD (1 if diagnosed with Alzheimer's disease) and MCI (1 if diagnosed with mild cognitive impairment). Of the $213$ individuals, $42$ have AD, $118$ have MCI and the other $53$ are cognitively normal (CN). The continuous covariates are age and education in years and the Mini-Mental State Exam (MMSE) score, where a higher score indicates better cognition. All covariates are standardized to have mean zero and variance one.

We also compare \ourmodel{} with two of its variants,
\begin{align}
    \log f_\theta(y\mid\bx,\s) &= g(y,w_{\bx,\s}) - \lpart_\theta(\bx,\s), \label{eqn: model without g2}\\
    \log f_\theta(y\mid\bx,\s) &= g_\s(y,w_{\bx,\s}) - \lpart_\theta(\bx,\s), \label{eqn: model with separate surfaces}
\end{align}
where $g$ and each $g_\s$ $(\s\in\mc{S})$ are represented as in \eqref{eqn: basis expansion g}. Model \eqref{eqn: model without g2} drops $g_2$, so the location affects the density only through $\beta(\s)$. Model \eqref{eqn: model with separate surfaces} fits a separate surface at each location. It contains \ourmodel{} as the special case $g_\s(y,w)=g_1(y,w)+g_2(y,\s)$, so it is an overparametrized version of \ourmodel{}, with $\K(\N+1)^2$ grid values instead of $2(\N+1)^2$, so its dimension grows linearly in $\K$. It also shares no information across locations.

We compare \ourmodel{} with the alternative models above, \ourgmodel{} of Section~\ref{ssec: generalized model}, the full grid LGP, the spLGP and LSBP of Section~\ref{sec: simulation} by $5$-fold cross-validation over individuals. The $213$ individuals are split at random into five folds, each model is fitted to four folds, and for each held-out individual $i$ we compute the log predictive density of all their responses, $\log\{T_{\iter}^{-1}\sum_{t=1}^{T_{\iter}}\prod_{\s\in\mc{S}}f_{\theta_t}(y_{i\s}\m\bx_i,\s)\}$, where $\theta_t$ $(t=1,\ldots,T_{\iter})$ are the posterior samples retained after burn-in. Its sum over all held-out individuals is larger for a model that predicts new individuals better. The individual is the unit here because the responses of one individual along the tract are dependent. We use the same $\N$, kernels and priors as in Section~\ref{sec: simulation}, with $2$ grid points per covariate for the full grid LGP, and for \ourgmodel{} we use $\N_\delta=5$, $(a_\delta,b_\delta)=(2,0.1)$ and the same Mat\'{e}rn kernel for $\Gamma_{y}$, $\Gamma_{\bx}$ and $\Gamma_{\s}$, with the Euclidean distance between covariate vectors for $\Gamma_{\bx}$. Every model is run for $22000$ iterations, the first $20000$ as burn-in, and the log predictive density uses the last $2000$. For \ourgmodel{}, $H$ is fixed at zero in the first half of the burn-in as a warm-up.
\begin{table}[htbp!]
    \centering
        \caption{Five-fold cross-validation over the $213$ individuals of the ADNI data. Log predictive density of all responses of each held-out individual, summed over individuals (larger is better, bold marks the largest).}
    \label{tab:cv}
    \begin{tabular}{|c|c|}
    \hline
       Model  &  Log predictive density \\
       \hline
\ourmodel{} (Eqn \eqref{eqn: model}) & 21451.6 \\
         \hline
   \ourgmodel{} (Eqn \eqref{eqn: generalized model}) & \textbf{21665.3} \\
         \hline
   Alternative 1 (Eqn \eqref{eqn: model without g2})   & 18349.3 \\
         \hline
            Alternative 2 (Eqn \eqref{eqn: model with separate surfaces}) & 21405.5 \\
         \hline
            Full grid LGP & 16651.0 \\
         \hline
            spLGP & 19797.9 \\
         \hline
            LSBP & 21312.9 \\
         \hline
    \end{tabular}
\end{table}

In Table \ref{tab:cv}, \ourgmodel{} has the largest log predictive density, followed by \ourmodel{}. Model \eqref{eqn: model with separate surfaces} comes close, although it uses $\K/2$ times as many grid values, and LSBP follows. The spLGP, the simpler model \eqref{eqn: model without g2} and the full grid LGP, whose covariate grid is inevitably coarse under $p=7$, predict new individuals clearly worse. The central $80\%$ and $95\%$ posterior predictive intervals of \ourgmodel{} contain $79.5\%$ and $94.5\%$ of the held-out responses. Since \ourgmodel{} fits the data best, we report its fit below. The posterior of $\sigma_\delta^2$ concentrates near zero, with median $0.0034$, so the densities in the figures below are computed from $g_1+g_2$ of the fit, with $\delta$ omitted.

\begin{figure}[t!]
\centering
\includegraphics[width=\textwidth]{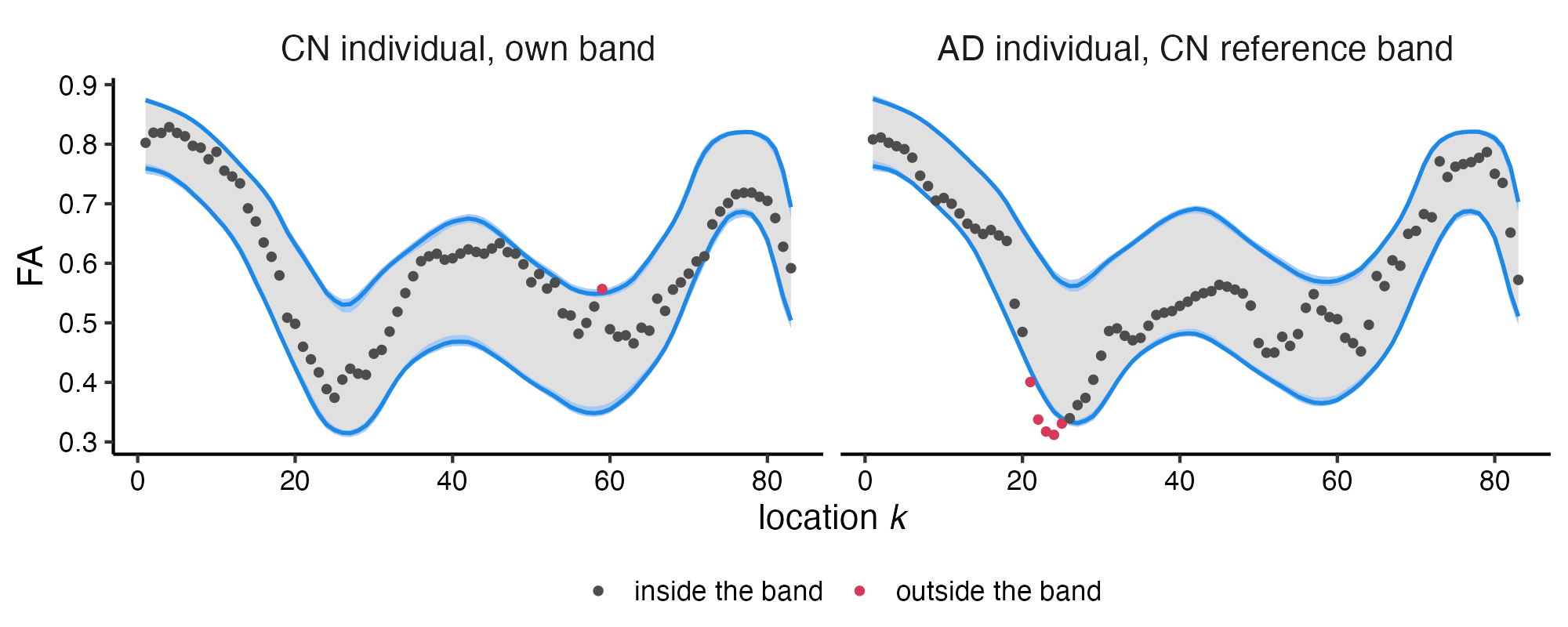}
\caption{$10$th and $90$th percentiles of FA along the tract (blue lines, posterior median, with pointwise $95\%$ credible bands). Left, a CN individual and their own band. Right, an AD individual of the same sex and the closest age and a CN reference band with the same sex, age, handedness and education. Points are the observed FA, red when outside the band. The credible bands are computed under model \eqref{eqn: conditional density model}.}
\label{fig:perc_band}
\end{figure}

We summarize the effect of diagnosis through $f_\theta$. The posterior median of the difference in the median of $f_\theta(\cdot\m\bx,\s)$ between an AD and a CN individual, holding all other covariates at their overall averages, is negative at every location, consistent with the known decline in FA associated with Alzheimer's disease \citep{sexton2011meta}. This difference remains negative at almost all locations across the five fits, each excluding one fifth of the individuals (Supplement~\ref{ssec: contrast}).

Figure \ref{fig:perc_band} shows the posterior median of the $10$th and $90$th percentiles of $f_\theta(\cdot\mid\bx,\s)$ along the tract, with the observed FA of two individuals. They are chosen to have the same sex and the closest age, while one is a CN and another is an AD individual. The band of the CN individual uses their own covariates. The band of the AD individual uses their own sex, age, handedness and education as a reference, with AD and MCI set to $0$ and MMSE at the CN group mean. Such a band gives a covariate-adjusted reference against which the FA profile of a new individual can be compared. \citet{zhang2021high} use p-value curves for a similar comparison. The FA of the AD individual falls below this reference band at locations $21$ to $25$ and stays in its lower part up to location $50$. Its use as a diagnostic tool would need validation against clinical outcomes. The pointwise credible bands rely on the conditional independence assumption in \eqref{eqn: conditional density model} and may therefore understate uncertainty when within-individual responses are strongly correlated along the tract (Supplement~\ref{ssec: dependence}).

\section{Conclusions}\label{sec: conclusion}
We proposed \ourmodel{}, a finite-rank LGP model for conditional densities that change with covariates and location. The normalizing constant of an LGP is the integral of an exponentiated GP, which has no closed form. We approximate the GP by linear interpolation on a regular grid, resulting in a log density that is piecewise linear in the response, and the likelihood, the mean and the quantiles available in closed form. We reduce dimensionality by expressing the log density, up to a normalizing term, as the sum of two bivariate functions: one of the response and a covariate index with location-dependent coefficients, following a single-index formulation, and the other of the response and location. The posterior contracts at the minimax rate of a bivariate $\alpha$-smooth function up to a logarithmic factor when the true log density is the sum of two such bivariate functions. Simulation studies revealed that \ourmodel{} was more accurate than a full grid LGP and spLGP with the same number of MCMC iterations, and comparable to or more accurate than LSBP. Application to the ADNI data generated covariate-adjusted percentile bands of FA along the corpus callosum. 

The current model has several limitations. Conditional independence across locations may understate posterior uncertainty when within-individual responses are strongly correlated along the tract. The single-index structure also restricts covariate effects at each location to one linear combination, while the posterior contraction result requires the true log density, up to normalization, to admit the specified additive bivariate representation. Future work will investigate multiple-index extensions \citep{xia2002adaptive} to accommodate richer covariate effects, and Gaussian copula models to capture within-individual dependence separately from the covariate-adjusted FA distributions at individual locations \citep{hoff2007extending}, retaining their closed-form likelihood contributions. Additional directions include establishing posterior contraction under dependent responses, developing simultaneous reference bands for tract-level FA profiles, and quantifying uncertainty in summaries such as the proportion of locations at which an individual's FA falls outside the reference band. The clinical utility of these summaries will require validation against clinical outcomes.

\if0\blind
{
\section*{Acknowledgement}
Data collection and sharing for this project was funded by the Alzheimer's Disease Neuroimaging Initiative (ADNI) (National Institutes of Health Grant U01 AG024904) and DOD ADNI (Department of Defense award number W81XWH-12-2-0012). ADNI is funded by the National Institute on Aging, the National Institute of Biomedical Imaging and Bioengineering, and through generous contributions from the following: AbbVie, Alzheimer's Association; Alzheimer's Drug Discovery Foundation; Araclon Biotech; BioClinica, Inc.; Biogen; Bristol-Myers Squibb Company; CereSpir, Inc.; Cogstate; Eisai Inc.; Elan Pharmaceuticals, Inc.; Eli Lilly and Company; EuroImmun; F. Hoffmann-La Roche Ltd and its affiliated company Genentech, Inc.; Fujirebio; GE Healthcare; IXICO Ltd.; Janssen Alzheimer Immunotherapy Research \& Development, LLC.; Johnson \& Johnson Pharmaceutical Research \& Development LLC.; Lumosity; Lundbeck; Merck \& Co., Inc.; Meso Scale Diagnostics, LLC.; NeuroRx Research; Neurotrack Technologies; Novartis Pharmaceuticals Corporation; Pfizer Inc.; Piramal Imaging; Servier; Takeda Pharmaceutical Company; and Transition Therapeutics. The Canadian Institutes of Health Research is providing funds to support ADNI clinical sites in Canada. Private sector contributions are facilitated by the Foundation for the National Institutes of Health (www.fnih.org). The grantee organization is the Northern California Institute for Research and Education, and the study is coordinated by the Alzheimer's Therapeutic Research Institute at the University of Southern California. ADNI data are disseminated by the Laboratory for Neuro Imaging at the University of Southern California. The authors also acknowledge partial research support from grants R01DE031134 and R21DE031879 from the United States National Institutes of Health. 
} \fi

\begin{singlespace} 
\bibliography{refs}   

\begin{thebibliography}{79}
\newcommand{\enquote}[1]{``#1''}
\providecommand{\natexlab}[1]{#1}
\providecommand{\url}[1]{\normalfont{#1}}
\providecommand{\urlprefix}{}

\bibitem[Adams, Murray, and MacKay(2009{\natexlab{a}})]{adams2009gaussian}
Adams, Ryan~Prescott, Iain Murray, and David J.~C. MacKay. 2009{\natexlab{a}}.
  ``The {G}aussian process density sampler.'' In \emph{Advances in Neural
  Information Processing Systems}, Vol.~21, 9--16.

\bibitem[Adams, Murray, and MacKay(2009{\natexlab{b}})]{adams2009tractable}
Adams, Ryan~Prescott, Iain Murray, and David J.~C. MacKay. 2009{\natexlab{b}}.
  ``Tractable nonparametric {B}ayesian inference in {P}oisson processes with
  {G}aussian process intensities.'' In \emph{Proceedings of the 26th
  International Conference on Machine Learning}, 9--16.

\bibitem[Adler(1990)]{adler1990introduction}
Adler, Robert~J. 1990. \emph{An Introduction to Continuity, Extrema, and
  Related Topics for General {G}aussian Processes}. Vol.~12 of \emph{Lecture
  Notes--Monograph Series}. Hayward, CA: Institute of Mathematical Statistics.

\bibitem[Bammer(2003)]{bammer2003basic}
Bammer, Roland. 2003. ``Basic principles of diffusion-weighted imaging.''
  \emph{European Journal of Radiology} 45 (3): 169--184.

\bibitem[Barrientos, Jara, and Quintana(2017)]{barrientos2017fully}
Barrientos, Andr{\'e}s~F., Alejandro Jara, and Fernando~A. Quintana. 2017.
  ``Fully nonparametric regression for bounded data using dependent {B}ernstein
  polynomials.'' \emph{Journal of the American Statistical Association} 112
  (518): 806--825.

\bibitem[Basser(1995)]{basser1995inferring}
Basser, Peter~J. 1995. ``Inferring microstructural features and the
  physiological state of tissues from diffusion-weighted images.'' \emph{NMR in
  Biomedicine} 8 (7): 333--344.

\bibitem[Beaumont, Zhang, and Balding(2002)]{beaumont2002approximate}
Beaumont, Mark~A., Wenyang Zhang, and David~J. Balding. 2002. ``Approximate
  {B}ayesian computation in population genetics.'' \emph{Genetics} 162 (4):
  2025--2035.

\bibitem[Chen(2017)]{chen2017tutorial}
Chen, Yen-Chi. 2017. ``A tutorial on kernel density estimation and recent
  advances.'' \emph{Biostatistics \& Epidemiology} 1 (1): 161--187.

\bibitem[Chen et~al.(2026)]{chen2026logistic}
Chen, Zichuan, Lucas Kock, Jeong~Eun Lee, and David~J. Nott. 2026. ``Logistic
  {G}aussian process density regression: a generalized {B}ayesian approach.''
  arXiv:2606.22915.

\bibitem[Dhara et~al.(2020)]{dhara2020new}
Dhara, Kumaresh, Stuart Lipsitz, Debdeep Pati, and Debajyoti Sinha. 2020. ``A
  new {B}ayesian single index model with or without covariates missing at
  random.'' \emph{Bayesian Analysis} 15 (3): 759--780.

\bibitem[Donner and Opper(2018)]{donner2018efficient}
Donner, Christian, and Manfred Opper. 2018. ``Efficient {B}ayesian inference
  for a {G}aussian process density model.'' In \emph{Proceedings of the 34th
  Conference on Uncertainty in Artificial Intelligence}, Monterey, CA. AUAI
  Press.

\bibitem[Dunson and Park(2008)]{dunson2008kernel}
Dunson, David~B., and Ju-Hyun Park. 2008. ``Kernel stick-breaking processes.''
  \emph{Biometrika} 95 (2): 307--323.

\bibitem[Dunson, Pillai, and Park(2007)]{dunson2007bayesian}
Dunson, David~B., Natesh Pillai, and Ju-Hyun Park. 2007. ``{B}ayesian density
  regression.'' \emph{Journal of the Royal Statistical Society: Series B
  (Statistical Methodology)} 69 (2): 163--183.

\bibitem[Fan, Yao, and Tong(1996)]{fan1996estimation}
Fan, Jianqing, Qiwei Yao, and Howell Tong. 1996. ``Estimation of conditional
  densities and sensitivity measures in nonlinear dynamical systems.''
  \emph{Biometrika} 83 (1): 189--206.

\bibitem[Gautier and Ginsbourger(2026)]{gautier2025continuous}
Gautier, Ath{\'e}na{\"i}s, and David Ginsbourger. 2026. ``Continuous logistic
  {G}aussian random measure fields for spatial distributional modelling.''
  \emph{Annals of the Institute of Statistical Mathematics} 1--39.

\bibitem[Gelfand, Kottas, and MacEachern(2005)]{gelfand2005bayesian}
Gelfand, Alan~E., Athanasios Kottas, and Steven~N. MacEachern. 2005.
  ``{B}ayesian nonparametric spatial modeling with {D}irichlet process
  mixing.'' \emph{Journal of the American Statistical Association} 100 (471):
  1021--1035.

\bibitem[Ghosal, Ghosh, and van~der Vaart(2000)]{ghosal2000convergence}
Ghosal, Subhashis, Jayanta~K. Ghosh, and Aad~W. van~der Vaart. 2000.
  ``Convergence rates of posterior distributions.'' \emph{The Annals of
  Statistics} 28 (2): 500--531.

\bibitem[Ghosal and van~der Vaart(2007{\natexlab{a}})]{ghosal2007convergence}
Ghosal, Subhashis, and Aad van~der Vaart. 2007{\natexlab{a}}. ``Convergence
  rates of posterior distributions for noniid observations.'' \emph{The Annals
  of Statistics} 35 (1): 192--223.

\bibitem[Ghosal and van~der Vaart(2007{\natexlab{b}})]{ghosal2007posterior}
Ghosal, Subhashis, and Aad van~der Vaart. 2007{\natexlab{b}}. ``Posterior
  convergence rates of {D}irichlet mixtures at smooth densities.'' \emph{The
  Annals of Statistics} 35 (2): 697--723.

\bibitem[Griebel and Hegland(2010)]{griebel2010finite}
Griebel, Michael, and Markus Hegland. 2010. ``A finite element method for
  density estimation with {G}aussian process priors.'' \emph{SIAM Journal on
  Numerical Analysis} 47 (6): 4759--4792.

\bibitem[Griffin and Steel(2006)]{griffin2006order}
Griffin, Jim~E., and Mark F.~J. Steel. 2006. ``Order-based dependent
  {D}irichlet processes.'' \emph{Journal of the American Statistical
  Association} 101 (473): 179--194.

\bibitem[Hastie and Tibshirani(1986)]{hastie1986generalized}
Hastie, Trevor, and Robert Tibshirani. 1986. ``Generalized additive models.''
  \emph{Statistical Science} 1 (3): 297--310.

\bibitem[Hegland(2007)]{hegland2007approximate}
Hegland, Markus. 2007. ``Approximate maximum a posteriori with {G}aussian
  process priors.'' \emph{Constructive Approximation} 26 (2): 205--224.

\bibitem[Hoff(2007)]{hoff2007extending}
Hoff, Peter~D. 2007. ``Extending the rank likelihood for semiparametric copula
  estimation.'' \emph{The Annals of Applied Statistics} 1 (1): 265--283.

\bibitem[Jack et~al.(2008)]{jack2008alzheimer}
Jack, Clifford~R., Jr., Matt~A. Bernstein, Nick~C. Fox, Paul Thompson, Gene
  Alexander, Danielle Harvey, Bret Borowski, et~al. 2008. ``The {A}lzheimer's
  {D}isease {N}euroimaging {I}nitiative ({ADNI}): {MRI} methods.''
  \emph{Journal of Magnetic Resonance Imaging} 27 (4): 685--691.

\bibitem[Jacobs(2014)]{jacobs2014nonparametric}
Jacobs, Justin~Wayne. 2014. ``Nonparametric density estimation on {R}iemannian
  manifolds.'' PhD diss., University of Maryland, Baltimore County.

\bibitem[Jiang(2007)]{jiang2007bayesian}
Jiang, Wenxin. 2007. ``{B}ayesian variable selection for high dimensional
  generalized linear models: convergence rates of the fitted densities.''
  \emph{The Annals of Statistics} 35 (4): 1487--1511.

\bibitem[Kalda and Siddiqui(2013)]{kalda2013nonparametric}
Kalda, Ankit, and Sikandar Siddiqui. 2013. ``Nonparametric conditional density
  estimation of short-term interest rate movements: procedures, results and
  risk management implications.'' \emph{Applied Financial Economics} 23 (8):
  671--684.

\bibitem[Kim, Bhattacharya, and Pati(2025)]{kim2025adaptive}
Kim, Jaehoan, Anirban Bhattacharya, and Debdeep Pati. 2025. ``Adaptive
  resolution for finite-rank {G}aussian processes.'' arXiv:2505.24066.

\bibitem[Kim and Mak(2026)]{kim2026isotonic}
Kim, Jaehoan, and Simon Mak. 2026. ``Isotonic surrogate modeling for computer
  experiments with many input variables.'' arXiv:2609.35354.

\bibitem[Klein et~al.(2015)]{klein2015bayesian}
Klein, Nadja, Thomas Kneib, Stefan Lang, and Alexander Sohn. 2015. ``{B}ayesian
  structured additive distributional regression with an application to regional
  income inequality in {G}ermany.'' \emph{The Annals of Applied Statistics} 9
  (2): 1024--1052.

\bibitem[Lenk(1988)]{lenk1988logistic}
Lenk, Peter~J. 1988. ``The logistic normal distribution for {B}ayesian,
  nonparametric, predictive densities.'' \emph{Journal of the American
  Statistical Association} 83 (402): 509--516.

\bibitem[Leonard(1978)]{leonard1978density}
Leonard, Tom. 1978. ``Density estimation, stochastic processes and prior
  information.'' \emph{Journal of the Royal Statistical Society: Series B
  (Methodological)} 40 (2): 113--132.

\bibitem[Li(1991)]{li1991sliced}
Li, Ker-Chau. 1991. ``Sliced inverse regression for dimension reduction.''
  \emph{Journal of the American Statistical Association} 86 (414): 316--327.

\bibitem[Luo, Zhu, and Zhu(2016)]{luo2016single}
Luo, Xinchao, Lixing Zhu, and Hongtu Zhu. 2016. ``Single-index varying
  coefficient model for functional responses.'' \emph{Biometrics} 72 (4):
  1275--1284.

\bibitem[Maatouk and Bay(2017)]{maatouk2017gaussian}
Maatouk, Hassan, and Xavier Bay. 2017. ``{G}aussian process emulators for
  computer experiments with inequality constraints.'' \emph{Mathematical
  Geosciences} 49 (5): 557--582.

\bibitem[Mak, Bingham, and Lu(2016)]{mak2016regional}
Mak, Simon, Derek Bingham, and Yi~Lu. 2016. ``A regional compound {P}oisson
  process for hurricane and tropical storm damage.'' \emph{Journal of the Royal
  Statistical Society: Series C (Applied Statistics)} 65 (5): 677--703.

\bibitem[McDonald and Campbell(2021)]{mcdonald2021review}
McDonald, Shaun, and David Campbell. 2021. ``A review of uncertainty
  quantification for density estimation.'' \emph{Statistics Surveys} 15: 1--71.

\bibitem[Murray, Adams, and MacKay(2010)]{murray2010elliptical}
Murray, Iain, Ryan~Prescott Adams, and David J.~C. MacKay. 2010. ``Elliptical
  slice sampling.'' In \emph{Proceedings of the Thirteenth International
  Conference on Artificial Intelligence and Statistics}, Vol.~9 of
  \emph{Proceedings of Machine Learning Research}, 541--548. PMLR.

\bibitem[Norets and Pati(2017)]{norets2017adaptive}
Norets, Andriy, and Debdeep Pati. 2017. ``Adaptive {B}ayesian estimation of
  conditional densities.'' \emph{Econometric Theory} 33 (4): 980--1012.

\bibitem[Norets and Pelenis(2014)]{norets2014posterior}
Norets, Andriy, and Justinas Pelenis. 2014. ``Posterior consistency in
  conditional density estimation by covariate dependent mixtures.''
  \emph{Econometric Theory} 30 (3): 606--646.

\bibitem[Odin, Bachoc, and Lagnoux(2026)]{odin2026contraction}
Odin, Elie, Fran{\c{c}}ois Bachoc, and Agn{\`e}s Lagnoux. 2026. ``Contraction
  rates and projection subspace estimation with {G}aussian process priors in
  high dimension.'' \emph{Bernoulli} 32 (2): 1645--1664.

\bibitem[{\"O}zarslan, Vemuri, and Mareci(2005)]{ozarslan2005generalized}
{\"O}zarslan, Evren, Baba~C. Vemuri, and Thomas~H. Mareci. 2005. ``Generalized
  scalar measures for diffusion {MRI} using trace, variance, and entropy.''
  \emph{Magnetic Resonance in Medicine} 53 (4): 866--876.

\bibitem[Paisley, Zhang, and Barr(2025)]{paisley2025gaussian}
Paisley, John, Wei Zhang, and Brian Barr. 2025. ``{G}aussian process tilted
  nonparametric density estimation using {F}isher divergence score matching.''
  In \emph{2025 IEEE International Conference on Big Data}, 1162--1171.

\bibitem[Pati, Dunson, and Tokdar(2013)]{pati2013posterior}
Pati, Debdeep, David~B. Dunson, and Surya~T. Tokdar. 2013. ``Posterior
  consistency in conditional distribution estimation.'' \emph{Journal of
  Multivariate Analysis} 116: 456--472.

\bibitem[Payne et~al.(2020)]{payne2020conditional}
Payne, Richard~D., Nilabja Guha, Yu~Ding, and Bani~K. Mallick. 2020. ``A
  conditional density estimation partition model using logistic {G}aussian
  processes.'' \emph{Biometrika} 107 (1): 173--190.

\bibitem[Petrone(1999)]{petrone1999random}
Petrone, Sonia. 1999. ``Random {B}ernstein polynomials.'' \emph{Scandinavian
  Journal of Statistics} 26 (3): 373--393.

\bibitem[Rasmussen and Williams(2006)]{rasmussen2006gaussian}
Rasmussen, Carl~Edward, and Christopher K.~I. Williams. 2006. \emph{{G}aussian
  Processes for Machine Learning}. Cambridge, MA: MIT Press.

\bibitem[Rau et~al.(2015)]{rau2015accurate}
Rau, Markus~Michael, Stella Seitz, Fabrice Brimioulle, Eibe Frank, Oliver
  Friedrich, Daniel Gruen, and Ben Hoyle. 2015. ``Accurate photometric redshift
  probability density estimation -- method comparison and application.''
  \emph{Monthly Notices of the Royal Astronomical Society} 452 (4): 3710--3725.

\bibitem[Ray, Pati, and Bhattacharya(2020)]{ray2020efficient}
Ray, Pallavi, Debdeep Pati, and Anirban Bhattacharya. 2020. ``Efficient
  {B}ayesian shape-restricted function estimation with constrained {G}aussian
  process priors.'' \emph{Statistics and Computing} 30 (4): 839--853.

\bibitem[Reich and Fuentes(2007)]{reich2007multivariate}
Reich, Brian~J., and Montserrat Fuentes. 2007. ``A multivariate semiparametric
  {B}ayesian spatial modeling framework for hurricane surface wind fields.''
  \emph{The Annals of Applied Statistics} 1 (1): 249--264.

\bibitem[Reich, Fuentes, and Dunson(2011)]{reich2011bayesian}
Reich, Brian~J., Montserrat Fuentes, and David~B. Dunson. 2011. ``{B}ayesian
  spatial quantile regression.'' \emph{Journal of the American Statistical
  Association} 106 (493): 6--20.

\bibitem[Rigby and Stasinopoulos(2005)]{rigby2005generalized}
Rigby, Robert~A., and D.~Mikis Stasinopoulos. 2005. ``Generalized additive
  models for location, scale and shape.'' \emph{Journal of the Royal
  Statistical Society: Series C (Applied Statistics)} 54 (3): 507--554.

\bibitem[Rigon and Durante(2021)]{rigon2021tractable}
Rigon, Tommaso, and Daniele Durante. 2021. ``Tractable Bayesian density
  regression via logit stick-breaking priors.'' \emph{Journal of Statistical
  Planning and Inference} 211: 131--142.

\bibitem[Riihim{\"a}ki and Vehtari(2014)]{riihimaki2014laplace}
Riihim{\"a}ki, Jaakko, and Aki Vehtari. 2014. ``{L}aplace approximation for
  logistic {G}aussian process density estimation and regression.''
  \emph{Bayesian Analysis} 9 (2): 425--448.

\bibitem[Roy, Ghosal, and Choudhury(2020)]{roy2020high}
Roy, Arkaprava, Subhashis Ghosal, and Kingshuk~Roy Choudhury. 2020. ``High
  dimensional single-index {B}ayesian modeling of brain atrophy.''
  \emph{Bayesian Analysis} 15 (4): 1229--1249.

\bibitem[Rue and Held(2005)]{rue2005gaussian}
Rue, H{\aa}vard, and Leonhard Held. 2005. \emph{{G}aussian {M}arkov Random
  Fields: Theory and Applications}. Vol. 104 of \emph{Monographs on Statistics
  and Applied Probability}. Boca Raton, FL: Chapman \& Hall/CRC.

\bibitem[Schaback(1995)]{schaback1995error}
Schaback, Robert. 1995. ``Error estimates and condition numbers for radial
  basis function interpolation.'' \emph{Advances in Computational Mathematics}
  3 (3): 251--264.

\bibitem[Sexton et~al.(2011)]{sexton2011meta}
Sexton, Claire~E., Ukwuori~G. Kalu, Nicola Filippini, Clare~E. Mackay, and
  Klaus~P. Ebmeier. 2011. ``A meta-analysis of diffusion tensor imaging in mild
  cognitive impairment and {A}lzheimer's disease.'' \emph{Neurobiology of
  Aging} 32 (12): 2322.e5--2322.e18.

\bibitem[Shen and Ghosal(2016)]{shen2016adaptive}
Shen, Weining, and Subhashis Ghosal. 2016. ``Adaptive {B}ayesian density
  regression for high-dimensional data.'' \emph{Bernoulli} 22 (1): 396--420.

\bibitem[Shen, Tokdar, and Ghosal(2013)]{shen2013adaptive}
Shen, Weining, Surya~T. Tokdar, and Subhashis Ghosal. 2013. ``Adaptive
  {B}ayesian multivariate density estimation with {D}irichlet mixtures.''
  \emph{Biometrika} 100 (3): 623--640.

\bibitem[Shen and Wasserman(2001)]{shen2001rates}
Shen, Xiaotong, and Larry Wasserman. 2001. ``Rates of convergence of posterior
  distributions.'' \emph{The Annals of Statistics} 29 (3): 687--714.

\bibitem[Smith et~al.(2006)]{smith2006tract}
Smith, Stephen~M., Mark Jenkinson, Heidi Johansen-Berg, Daniel Rueckert,
  Thomas~E. Nichols, Clare~E. Mackay, Kate~E. Watkins, et~al. 2006.
  ``Tract-based spatial statistics: voxelwise analysis of multi-subject
  diffusion data.'' \emph{NeuroImage} 31 (4): 1487--1505.

\bibitem[Tokdar(2007)]{tokdar2007towards}
Tokdar, Surya~T. 2007. ``Towards a faster implementation of density estimation
  with logistic {G}aussian process priors.'' \emph{Journal of Computational and
  Graphical Statistics} 16 (3): 633--655.

\bibitem[Tokdar(2011)]{tokdar2011dimension}
Tokdar, Surya~T. 2011. ``Dimension adaptability of {G}aussian process models
  with variable selection and projection.'' arXiv:1112.0716.

\bibitem[Tokdar, Zhu, and Ghosh(2010)]{tokdar2010bayesian}
Tokdar, Surya~T., Yu~M. Zhu, and Jayanta~K. Ghosh. 2010. ``{B}ayesian density
  regression with logistic {G}aussian process and subspace projection.''
  \emph{Bayesian Analysis} 5 (2): 319--344.

\bibitem[Tsybakov(2009)]{tsybakov2004introduction}
Tsybakov, Alexandre~B. 2009. \emph{Introduction to Nonparametric Estimation}.
  Springer Series in Statistics. New York: Springer.

\bibitem[van~der Vaart and van Zanten(2008{\natexlab{a}})]{van2008rates}
van~der Vaart, Aad~W., and J.~Harry van Zanten. 2008{\natexlab{a}}. ``Rates of
  contraction of posterior distributions based on {G}aussian process priors.''
  \emph{The Annals of Statistics} 36 (3): 1435--1463.

\bibitem[van~der Vaart and van Zanten(2008{\natexlab{b}})]{van2008reproducing}
van~der Vaart, Aad~W., and J.~Harry van Zanten. 2008{\natexlab{b}}.
  ``Reproducing kernel {H}ilbert spaces of {G}aussian priors.'' In
  \emph{Pushing the Limits of Contemporary Statistics: Contributions in Honor
  of Jayanta K. Ghosh},  edited by Bertrand Clarke and Subhashis Ghosal, Vol.~3
  of \emph{IMS Collections}, 200--222. Beachwood, OH: Institute of Mathematical
  Statistics.

\bibitem[van~der Vaart and van Zanten(2009)]{van2009adaptive}
van~der Vaart, Aad~W., and J.~Harry van Zanten. 2009. ``Adaptive {B}ayesian
  estimation using a {G}aussian random field with inverse {G}amma bandwidth.''
  \emph{The Annals of Statistics} 37 (5B): 2655--2675.

\bibitem[Verdinelli and Wasserman(1998)]{verdinelli1998bayesian}
Verdinelli, Isabella, and Larry Wasserman. 1998. ``{B}ayesian goodness-of-fit
  testing using infinite-dimensional exponential families.'' \emph{The Annals
  of Statistics} 26 (4): 1215--1241.

\bibitem[Vershynin(2026)]{vershynin2026high}
Vershynin, Roman. 2026. \emph{High-Dimensional Probability: An Introduction
  with Applications in Data Science}. 2nd ed. Cambridge: Cambridge University
  Press.

\bibitem[Wade et~al.(2014)]{wade2014improving}
Wade, Sara, David~B. Dunson, Sonia Petrone, and Lorenzo Trippa. 2014.
  ``Improving prediction from {D}irichlet process mixtures via enrichment.''
  \emph{Journal of Machine Learning Research} 15: 1041--1071.

\bibitem[Wang(2024)]{wang2024nonparametric}
Wang, Haoxuan. 2024. ``Nonparametric {B}ayesian density estimation with
  {G}aussian processes.'' Master's thesis, Duke University.

\bibitem[Wendland(2004)]{wendland2004scattered}
Wendland, Holger. 2004. \emph{Scattered Data Approximation}. Cambridge:
  Cambridge University Press.

\bibitem[Wood(2010)]{wood2010statistical}
Wood, Simon~N. 2010. ``Statistical inference for noisy nonlinear ecological
  dynamic systems.'' \emph{Nature} 466 (7310): 1102--1104.

\bibitem[Xia et~al.(2002)]{xia2002adaptive}
Xia, Yingcun, Howell Tong, W.~K. Li, and Li-Xing Zhu. 2002. ``An adaptive
  estimation of dimension reduction space.'' \emph{Journal of the Royal
  Statistical Society: Series B (Statistical Methodology)} 64 (3): 363--410.

\bibitem[Zhang(2004)]{zhang2004inconsistent}
Zhang, Hao. 2004. ``Inconsistent estimation and asymptotically equal
  interpolations in model-based geostatistics.'' \emph{Journal of the American
  Statistical Association} 99 (465): 250--261.

\bibitem[Zhang et~al.(2022)]{zhang2021high}
Zhang, Zhengwu, Xiao Wang, Linglong Kong, and Hongtu Zhu. 2022.
  ``High-dimensional spatial quantile function-on-scalar regression.''
  \emph{Journal of the American Statistical Association} 117 (539): 1563--1578.

\end{thebibliography}
\end{singlespace}


\clearpage
\setcounter{section}{0}
\setcounter{equation}{0}
\renewcommand{\thesection}{\Alph{section}}
\renewcommand{\theequation}{S\arabic{equation}}
\renewcommand{\theHsection}{supp.\Alph{section}}
\renewcommand{\theHequation}{supp.\arabic{equation}}
\counterwithin{figure}{section}
\counterwithin{table}{section}
\renewcommand{\theHfigure}{supp.\thefigure}
\renewcommand{\theHtable}{supp.\thetable}
\titleformat{\section}[block]{\normalfont\Large\bfseries}{Supplement~\thesection:}{0.5em}{}
\titleformat{\subsection}[block]{\normalfont\bfseries}{\thesubsection}{0.5em}{}

\begin{center}
{\LARGE\bf Supplementary Materials}
\end{center}

Supplement \ref{sec: closed form} gives the closed form calculations. Supplement \ref{sec: proof of theorem} derives the detailed proof of Theorem \ref{thm:post_cont}. Supplement \ref{sec: numerical details} presents the benchmark methods, the coverage of the credible band, a check of dependence across locations and the posterior contrast in the ADNI data. 

\section{Closed form calculations}\label{sec: closed form}

We derive the normalizing constant, the mean and the quantiles of $f_\theta(\cdot\m\bx,\s)$ in \eqref{eqn: model}. Fix $(\bx,\s)$, and let $u_j=j/\N$ $(j=0,\ldots,\N)$ be the knots and $\Delta_\N=1/\N$. By \eqref{eqn: basis expansion g}, the exponent of $f_\theta$ is
\begin{equation}\label{eqn: normalizing constant}
    g_1(y,w_{\bx,\s})+g_2(y,\s)=\sum_{i=0}^{\N}\delta_i\,\phi_i(y), \qquad \delta_i=\sum_{j=0}^{\N}\big\{\xi^{(1)}_{ij}\,\phi_j(w_{\bx,\s})+\xi^{(2)}_{ij}\,\phi_j(\s)\big\},
\end{equation}
which is piecewise linear in $y$ with value $\delta_j$ at $u_j$. On $[u_j,u_{j+1}]$ it equals $\gamma_j+\lambda_j y$, with slope $\lambda_j=(\delta_{j+1}-\delta_j)/\Delta_\N$ and intercept $\gamma_j=\delta_j-\lambda_j u_j$ $(j=0,\ldots,\N-1)$. For \ourgmodel{} in \eqref{eqn: generalized model}, the exponent gains $\sum_{k=0}^{\N_\delta}h_k(\bx,\s)\,\phi^{\delta}_k(y)$. This term is piecewise linear with knots $k/\N_\delta$, so the exponent is piecewise linear with knots at the union of $\{u_j\}$ and $\{k/\N_\delta\}$, and the calculations below hold with the $u_j$ replaced by the knots of this union and the $\delta_j$ by the values of the exponent there. If $\N_\delta$ divides $\N$, the union is $\{u_j\}$, and every $\delta_i$ is replaced by $\delta_i+\sum_{k}h_k(\bx,\s)\,\phi^{\delta}_k(u_i)$.

\subsection{Calculation of the normalizing constant}\label{ssec: normalizing constant}
Integrating over each cell and using $\gamma_j+\lambda_j u_j=\delta_j$ and $\gamma_j+\lambda_j u_{j+1}=\delta_{j+1}$,
\begin{align}
    \exp\{\lpart_\theta(\bx,\s)\} &= \sum_{j=0}^{\N-1}\exp(\gamma_j)\int_{u_j}^{u_{j+1}}\exp(\lambda_j y)\,dy \label{eqn: integral over cells}\\
    &= \sum_{j=0}^{\N-1}\frac{\exp(\delta_{j+1})-\exp(\delta_j)}{\lambda_j}, \label{eqn: normalizing constant closed form}
\end{align}
where a summand with $\lambda_j=0$ is its limit $\Delta_\N\exp(\delta_j)$. In the implementation, we subtract $\max_j \delta_j$ from every $\delta_j$ before exponentiation and add it back after taking the logarithm, which avoids overflow. Since at most two of the $\phi_j(w_{\bx,\s})$ and at most two of the $\phi_j(\s)$ are nonzero, each $\delta_i$ combines at most four coefficients, and $\lpart_\theta(\bx,\s)$ costs $\bigO(\N)$ operations.

\subsection{Calculation of Mean}\label{ssec: mean}
By \eqref{eqn: integral over cells} and integration by parts,
\begin{align*}
    \bbE(y\m\bx,\s) &= \exp\{-\lpart_\theta(\bx,\s)\}\sum_{j=0}^{\N-1}\exp(\gamma_j)\int_{u_j}^{u_{j+1}}y\exp(\lambda_j y)\,dy\\
    &= \exp\{-\lpart_\theta(\bx,\s)\}\sum_{j=0}^{\N-1}\bigg\{\Big(\frac{u_{j+1}}{\lambda_j}-\frac{1}{\lambda_j^2}\Big)e^{\delta_{j+1}}-\Big(\frac{u_j}{\lambda_j}-\frac{1}{\lambda_j^2}\Big)e^{\delta_j}\bigg\},
\end{align*}
where a summand with $\lambda_j=0$ is replaced by $e^{\delta_j}(u_{j+1}^2-u_j^2)/2$.

\subsection{Calculation of quantiles}\label{ssec: quantiles}
Let $t_\eta$ be the $\eta$ quantile of $f_\theta(\cdot\m\bx,\s)$ for $\eta\in(0,1)$, and let $S_j$ be the sum of the first $j$ summands in \eqref{eqn: normalizing constant closed form}, with $S_0=0$. Let $j_0\in\{0,\ldots,\N-1\}$ be the largest index with $S_{j_0}\le\eta\exp\{\lpart_\theta(\bx,\s)\}$, so that $u_{j_0}\le t_\eta<u_{j_0+1}$. By \eqref{eqn: integral over cells}, $t_\eta$ solves
\begin{equation*}
    S_{j_0}+\exp(\gamma_{j_0})\int_{u_{j_0}}^{t_\eta}\exp(\lambda_{j_0}y)\,dy=\eta\exp\{\lpart_\theta(\bx,\s)\}.
\end{equation*}
For $\lambda_{j_0}\neq0$ the integral equals $\{\exp(\gamma_{j_0}+\lambda_{j_0}t_\eta)-\exp(\delta_{j_0})\}/\lambda_{j_0}$, so
\begin{equation*}
    t_\eta=\frac{\log R_\eta-\gamma_{j_0}}{\lambda_{j_0}}, \qquad R_\eta=\lambda_{j_0}\big[\eta\exp\{\lpart_\theta(\bx,\s)\}-S_{j_0}\big]+\exp(\delta_{j_0}).
\end{equation*}
For $\lambda_{j_0}=0$, $t_\eta=u_{j_0}+e^{-\delta_{j_0}}\big[\eta\exp\{\lpart_\theta(\bx,\s)\}-S_{j_0}\big]$.

\section{Proof of Theorem \ref{thm:post_cont}}\label{sec: proof of theorem}

The proof verifies the three conditions of Theorem 2.1 in \cite{ghosal2000convergence}. Lemma \ref{lemma:hellinger} bounds the Hellinger distance, the Kullback--Leibler divergence and the Kullback--Leibler variation by the sup norm distance between log densities. This distance is split into the error of the Gaussian process $G_l$ in approximating $g^{(0)}_l$, the error of its linear interpolation $g_l$ in approximating $G_l$, and the error caused by $\beta \neq \beta^{(0)}$. These are controlled by Lemma \ref{lemma:gp}, Lemma \ref{lemma:interp} and the Gaussian prior on $\beta$, respectively. The interpolation error must be negligible compared with the small ball probability of the Gaussian process, which determines how fast $\N$ must grow. Section \ref{ssec: reduction} states the three conditions, Section \ref{ssec: auxiliary lemmas} collects the lemmas, Section \ref{ssec: choice of rates} fixes $\varepsilon_n$, $\kappa_n$ and $\N$, and Sections \ref{ssec: sieve condition} to \ref{ssec: KL condition} verify the conditions.

Let $B=\max_{\s\in\mc{S}}\|\beta^{(0)}(\s)\|_2$, so that $\bx^\T\beta^{(0)}(\s)\in[-B,B]$ for $\|\bx\|_2\le1$. So $g^{(0)}_1$ enters $\log\Upsilon_0$ only on $[0,1]\times[-B,B]$, where it is $\alpha$-H\"older. If $B=0$, we replace $g^{(0)}_2$ by $g^{(0)}_1(\cdot,0)+g^{(0)}_2$, $g^{(0)}_1$ by $0$ and each $\beta^{(0)}(\s)$ by a unit vector, so that $B=1$. With $\tilde\beta^{(0)}(\s)=\beta^{(0)}(\s)/B$, $w^{(0)}_{\bx,\s}=(\bx^\T\tilde\beta^{(0)}(\s)+1)/2\in[0,1]$ and $\tilde g^{(0)}_1(y,w)=g^{(0)}_1(y,B(2w-1))\in\bbC^\alpha[0,1]^2$, we have $g^{(0)}_1(y,\bx^\T\beta^{(0)}(\s))=\tilde g^{(0)}_1(y,w^{(0)}_{\bx,\s})$ and $\|\tilde\beta^{(0)}(\s)\|_2\le1$. We write $g^{(0)}_1$ and $\beta^{(0)}$ for $\tilde g^{(0)}_1$ and $\tilde\beta^{(0)}$ below.

\subsection{Reduction to three conditions}\label{ssec: reduction}
Write $Z_i = \{\bx_i, (y_{i\s})_{\s\in\mc{S}}\}$. By the sampling assumption in Section~\ref{sec: posterior contraction}, $Z_1, \ldots, Z_n$ are i.i.d.\ with density
\begin{equation*}
    q_\theta(z) = \prod_{\s\in\mc{S}} f_\theta(y_\s \m \bx, \s), \qquad z = \{\bx, (y_\s)_{\s\in\mc{S}}\},
\end{equation*}
with respect to $\mc{H}_0 \otimes \mathrm{Leb}^{\K}$, and the true density $q_0$ is defined analogously with $f_0$. Let $\Hel$ denote the Hellinger distance between two densities. For fixed $\bx$, let $a_\s(\bx) = \int \{f_{\theta_1}(y\m\bx,\s)f_{\theta_2}(y\m\bx,\s)\}^{1/2}dy \in [0,1]$. Then $d^2_h(f_{\theta_1}, f_{\theta_2}) = 2\K^{-1}\sum_\s\int\{1 - a_\s\}\,d\mc{H}_0$ and $\Hel^2(q_{\theta_1}, q_{\theta_2}) = 2\int\{1 - \prod_\s a_\s\}\,d\mc{H}_0$. Since $\max_\s(1-a_\s) \le 1 - \prod_\s a_\s \le \sum_\s(1-a_\s)$,
\begin{equation}\label{eqn: hellinger equivalence}
    d^2_h(f_{\theta_1}, f_{\theta_2}) \le \Hel^2(q_{\theta_1}, q_{\theta_2}) \le \K\, d^2_h(f_{\theta_1}, f_{\theta_2}).
\end{equation}
Define
\begin{align*}
    \mrm{KL}(f_0,f) &= \K^{-1}\sum_{\s\in\mc{S}}\int \log \p{f_0(y\mid \bx,\s)/f(y\mid \bx,\s)}\,f_0(y\mid \bx,\s)\,dy\,d\mc{H}_0(\bx),\\
    \mrm{V}(f_0,f) &= \K^{-1}\sum_{\s\in\mc{S}}\int \log^2 \p{f_0(y\mid \bx,\s)/f(y\mid \bx,\s)}\,f_0(y\mid \bx,\s)\,dy\,d\mc{H}_0(\bx).
\end{align*}
Since $\log(q_0/q_\theta)(z) = \sum_\s \log(f_0/f_\theta)(y_\s\m\bx,\s)$, we have $\mrm{KL}(q_0, q_\theta) = \K\,\mrm{KL}(f_0, f_\theta)$, and the Cauchy--Schwarz inequality gives $\mrm{V}(q_0, q_\theta) \le \K^2\,\mrm{V}(f_0, f_\theta)$. It suffices to find constants $A, C' > 0$ and, for every $C > 0$, a sieve $\mc{F}_n \subset \mc{F}$ such that for all large $n$
\begin{align}
 \Pi(\mc{F}^c_n) &\le \exp\p{-C\,n\varepsilon^2_n},\label{eqn: sieve condition}\\
\log \mc{C}(\varepsilon_n,\mc{F}_n, d_h) &\le n\varepsilon^2_n\label{eqn: entropy condition},\\
\Pi(f: \mrm{KL}(f_0,f)\leq A\varepsilon^2_n, \mrm{V}(f_0,f)\leq A\varepsilon^2_n) &\ge \exp{(-\,C'n\varepsilon^2_n)}\label{eqn: KL condition},
\end{align}
where $\Pi$ is the prior and $\mc{C}(\varepsilon, \cdot, d)$ is the $\varepsilon$ covering number. We call these the sieve condition, the metric entropy condition and the KL condition. By \eqref{eqn: hellinger equivalence} and the relations for $\mrm{KL}$ and $\mrm{V}$ above, these conditions with the sieve taken for $C = (C'+4)L^2$ and $L = 2\K\max(1,\sqrt{A})$ give (2.2) to (2.4) of \cite{ghosal2000convergence} for $q_\theta$ at the rate $L\varepsilon_n$, which satisfies $L\varepsilon_n\to0$ and $nL^2\varepsilon_n^2\to\infty$. Theorem 2.1 there gives contraction of $q_\theta$ in $\Hel$ at this rate in probability, and hence of $f_\theta$ in $d_h$, since $d_h \le \Hel$. Since the posterior probability is at most one, this gives \eqref{eqn: posterior contraction}.

We represent the prior through Gaussian processes. Since $\sigma_l^2$ is fixed, we set $\sigma_l^2=1$ without loss of generality, as a fixed scale only changes the constants below. For fixed $\kappa>0$, let $G^\kappa: [0, 1]^2 \to \bbR$ be a centered Gaussian process with covariance kernel $K(z, z') = \exp(-\kappa^2\,\|z - z'\|_2^2)$. Given $\kappa_l$, let $G_l$ be distributed as $G^{\kappa_l}$ $(l=1,2)$, with $(G_1,\kappa_1)$, $(G_2,\kappa_2)$ and $\beta$ independent. Let $T(g) := \big(g(u_i, u_j)\big)_{0\le i,j\le \N}$. Then $T(G_l) \sim \cN(0, \Gamma)$ given $\kappa_l$ \citep[see Proposition 3 of][]{maatouk2017gaussian}, so we take $\xi^{(1)} = T(G_1)$ and $\xi^{(2)} = T(G_2)$. Then the functions $g_1, g_2$ in \eqref{eqn: model} are the linear interpolations of $G_1, G_2$ on the grid,
\begin{equation}\label{eqn: approximation to original GP}
    g_l (z_1, z_2) = \sum_{i, j=0}^\N T(G_l)_{ij}\, \phi_i (z_1)\phi_j(z_2) \quad (l = 1, 2).
\end{equation}
For a bounded measurable function $v$ on $[0,1]$, let $p_v = e^{v}/\int_0^1 e^{v(y)}dy$. Then $f_\theta(\cdot \m \bx, \s) = p_{v_\theta(\cdot, \bx, \s)}$ and $f_0(\cdot \m \bx, \s) = p_{v_0(\cdot, \bx, \s)}$ with the log densities
\begin{equation*}
    v_\theta(y, \bx, \s) = g_1(y,\,w_{\bx,\,\s}) + g_2(y,\,\s), \qquad
    v_0(y, \bx, \s) = g^{(0)}_1(y,\,w^{(0)}_{\bx,\,\s}) + g^{(0)}_2(y,\,\s).
\end{equation*}

\subsection{Auxiliary lemmas}\label{ssec: auxiliary lemmas}
For $\xi \in \bbR^{(\N+1)^2}$, write $g_\xi(z_1, z_2) = \sum^\N_{i,j=0}\xi_{ij}\phi_i(z_1)\phi_j(z_2)$, where the $\phi_j$ in \eqref{eqn: hat basis} are defined on the whole real line. Then $g_1 = g_{\xi^{(1)}}$ and $g_2 = g_{\xi^{(2)}}$.
\begin{lemma}\label{lemma:hellinger}
\begin{itemize}
    \item[(i)] For all $y \in [0,1]$ and $z, z' \in \bbR$,
    \begin{equation*}
        |g_\xi(y, z) - g_{\xi'}(y, z)| \le \|\xi - \xi'\|_\infty, \qquad |g_\xi(y, z) - g_\xi(y, z')| \le 4\N\|\xi\|_\infty|z - z'|.
    \end{equation*}
    \item[(ii)] For bounded measurable functions $v, v'$ on $[0,1]$, $\Hel(p_v, p_{v'}) \le \|v - v'\|_\infty e^{\|v - v'\|_\infty/2}$. If $\|v - v'\|_\infty \le 1$, then $\mrm{KL}(p_v,p_{v'}) \le C_0\|v-v'\|^2_{\infty}$ and $\mrm{V}(p_v,p_{v'}) \le C_0\|v-v'\|^2_{\infty}$ for a universal constant $C_0 > 0$.
    \item[(iii)] For $\tilde\theta = (\tilde\xi^{(1)}, \tilde\xi^{(2)}, \tilde\beta)$ and $\theta = (\xi^{(1)}, \xi^{(2)}, \beta)$ in $\Theta$, let
    \begin{equation*}
        \dTheta(\tilde\theta, \theta) := \|\tilde\xi^{(1)} - \xi^{(1)}\|_\infty + \|\tilde\xi^{(2)} - \xi^{(2)}\|_\infty + 2\N\|\xi^{(1)}\|_\infty\|\tilde\beta - \beta\|_2.
    \end{equation*}
    Then $\sup_{y, \bx, \s}|v_{\tilde\theta} - v_{\theta}| \le \dTheta(\tilde\theta, \theta)$ and $d_h(f_{\tilde\theta}, f_{\theta}) \le \dTheta(\tilde\theta, \theta)\, e^{\dTheta(\tilde\theta, \theta)/2}$.
\end{itemize}
Here and below, the supremum is over $y \in [0,1]$, $\|\bx\|_2 \le 1$ and $\s \in \mc{S}$.
\end{lemma}
\begin{proof}
(i) The basis functions satisfy $\phi_j \ge 0$ and $\sum_{j=0}^\N \phi_j(z) \le 1$ for all $z \in \bbR$, with equality on $[0,1]$. Hence $|g_\xi(y, z) - g_{\xi'}(y, z)| \le \|\xi - \xi'\|_\infty \sum_{i}\phi_i(y)\sum_j \phi_j(z) \le \|\xi - \xi'\|_\infty$. Each $\phi_j$ is piecewise linear with slopes $0$ or $\pm 1/\Delta_\N$, so $|\phi_j(z) - \phi_j(z')| \le \N|z - z'|$. At most two of $\{\phi_j(z)\}_j$ and at most two of $\{\phi_j(z')\}_j$ are nonzero. Therefore,
\begin{equation*}
    |g_\xi(y, z) - g_\xi(y, z')| \le \sum_{i} \phi_i(y) \sum_{j}|\xi_{ij}|\,|\phi_j(z) - \phi_j(z')| \le 4\N\|\xi\|_\infty|z - z'|.
\end{equation*}
(ii) Lemma 3.1 in \cite{van2008rates} gives the Hellinger bound. It also gives $\mrm{KL}(p_v,p_{v'}) \lesssim \|v-v'\|^2_\infty e^{\|v-v'\|_\infty}(1 + \|v-v'\|_\infty)$ and $\mrm{V}(p_v,p_{v'}) \lesssim \|v-v'\|^2_\infty e^{\|v-v'\|_\infty}(1 + \|v-v'\|_\infty)^2$. Both right hand sides are at most a constant times $\|v-v'\|^2_\infty$ when $\|v-v'\|_\infty \le 1$.

(iii) Let $\tilde w = c((\bx^\T\tilde\beta(\s) + 1)/2)$ and $w = c((\bx^\T\beta(\s) + 1)/2)$. Since $c$ is $1$-Lipschitz and $\|\bx\|_2 \le 1$, $|\tilde w - w| \le \|\tilde\beta(\s) - \beta(\s)\|_2/2 \le \|\tilde\beta - \beta\|_2/2$. By (i), for all $(y, \bx, \s)$,
\begin{align*}
    |v_{\tilde\theta} - v_{\theta}| &\le |g_{\tilde\xi^{(1)}}(y, \tilde w) - g_{\xi^{(1)}}(y, \tilde w)| + |g_{\xi^{(1)}}(y, \tilde w) - g_{\xi^{(1)}}(y, w)| + |g_{\tilde\xi^{(2)}}(y, \s) - g_{\xi^{(2)}}(y, \s)| \\
    &\le \dTheta(\tilde\theta, \theta).
\end{align*}
Applying (ii) for each fixed $(\bx, \s)$, squaring, integrating with respect to $\mc{H}_0$ and averaging over $\s$ yields $d^2_h(f_{\tilde\theta}, f_{\theta}) \le \dTheta^2(\tilde\theta, \theta)\, e^{\dTheta(\tilde\theta, \theta)}$.
\end{proof}

In Lemmas \ref{lemma:gp} and \ref{lemma:interp}, $\kappa$ is fixed, $G_1$ stands for $G^\kappa$ and $g_1$ for its interpolation \eqref{eqn: approximation to original GP}. Let $\mathbb{H}^{\kappa}$ be the reproducing kernel Hilbert space (RKHS) of $G_1$ and $\mathbb{H}^{\kappa}_1$ its unit ball. Let $\mathbb{B}_1$ be the unit ball of $\big(C[0,1]^2, \|\cdot\|_\infty\big)$ and $\Phi$ the standard normal distribution function.
\begin{lemma}\label{lemma:gp}
Let $g^{(0)} \in \bbC^{\alpha}[0,1]^2$. There exist constants $\kappa_0, \epsilon_0, A_1, C_\alpha > 0$ depending only on $g^{(0)}$ and $\alpha$ such that the following hold for all $\kappa \ge \kappa_0$ and $\epsilon \le \epsilon_0$.
\begin{itemize}
    \item[(a)] $-\log \Pi(\|G_1\|_\infty < \epsilon) \le A_1\kappa^2\{\log(\kappa/\epsilon)\}^3$.
    \item[(b)] If $\epsilon \ge 2C_\alpha\kappa^{-\alpha}$, then $-\log \Pi(\|G_1 - g^{(0)}\|_\infty < \epsilon) \le A_1\kappa^2\{\log(\kappa/\epsilon)\}^3$.
    \item[(c)] $\log \mc{C}(\epsilon, \mathbb{H}^{\kappa}_1, \|\cdot\|_\infty) \le A_1\kappa^2\{\log(1/\epsilon)\}^3$.
    \item[(d)] $\|\zeta\|_\infty \le 1$ for every $\zeta \in \mathbb{H}^{\kappa}_1$.
    \item[(e)] For every $Q \ge 0$, $\Pi(G_1 \notin Q\mathbb{H}^{\kappa}_1 + \epsilon\mathbb{B}_1) \le 1 - \Phi\big(\Phi^{-1}(\Pi(\|G_1\|_\infty < \epsilon)) + Q\big)$.
    \item[(f)] For $0<\kappa\le r$, $\mathbb{H}^{\kappa}_1 \subset (r/\kappa)\,\mathbb{H}^{r}_1$ and $\Pi(\|G^\kappa\|_\infty < \epsilon) \ge \Pi(\|G^r\|_\infty < \epsilon)$.
\end{itemize}
\end{lemma}
\begin{proof}
The process $G_1$ equals the rescaled process $z \mapsto W_{\kappa z}$ in \cite{van2009adaptive}, where $W$ has covariance kernel $\exp(-\|z-z'\|_2^2)$. Its spectral measure has a Gaussian density, so it satisfies condition (3.3) of \cite{van2009adaptive} and its density is bounded away from zero near the origin. We take $\kappa_0\ge1$ and $\epsilon_0<1/2$. Part (a) follows from Lemma 4.6 at $\epsilon/2$, since $\Pi(\|G_1\|_\infty<\epsilon)\ge\Pi(\|G_1\|_\infty\le\epsilon/2)$ and $\log(2\kappa/\epsilon)\le2\log(\kappa/\epsilon)$, after enlarging $A_1$, and part (c) is Lemma 4.5 of \cite{van2009adaptive} with $d = 2$. By Lemma 4.3 of \cite{van2009adaptive}, there exist constants $C_\alpha, D > 0$ such that $\inf\{\|\zeta\|^2_{\mathbb{H}^\kappa}: \|\zeta - g^{(0)}\|_\infty \le C_\alpha\kappa^{-\alpha}\} \le D\kappa^2$. For $\zeta\in\mathbb{H}^\kappa$ with $\|\zeta-g^{(0)}\|_\infty\le\epsilon/2$, the triangle inequality and Lemma 5.2 of \cite{van2008reproducing} with the symmetric set $\{\|\cdot\|_\infty<\epsilon/2\}$ give $\Pi(\|G_1-g^{(0)}\|_\infty<\epsilon)\ge\Pi(\|G_1-\zeta\|_\infty<\epsilon/2)\ge e^{-\|\zeta\|^2_{\mathbb{H}^\kappa}/2}\,\Pi(\|G_1\|_\infty<\epsilon/2)$. Hence
\begin{equation*}
    -\log \Pi(\|G_1 - g^{(0)}\|_\infty < \epsilon) \le \inf_{\|\zeta - g^{(0)}\|_\infty \le \epsilon/2}\frac{1}{2}\|\zeta\|^2_{\mathbb{H}^\kappa} - \log \Pi(\|G_1\|_\infty < \epsilon/2).
\end{equation*}
If $\epsilon/2 \ge C_\alpha\kappa^{-\alpha}$, the first term is at most $D\kappa^2/2$, and part (a) bounds the second term. Enlarging $A_1$ gives (b). Part (d) follows from the reproducing property, since $|\zeta(z)| \le \|\zeta\|_{\mathbb{H}^\kappa}K(z,z)^{1/2} = \|\zeta\|_{\mathbb{H}^\kappa}$. Part (e) is Borell's inequality in Theorem 5.1 of \cite{van2008reproducing}. The first claim of (f) is Lemma 4.7 of \cite{van2009adaptive} with $d=2$, which applies since the Gaussian spectral density is radially decreasing. For the second, $\|G^\kappa\|_\infty = \sup_{t\in[0,\kappa]^2}|W_t|$ is nondecreasing in $\kappa$. Parts (d), (e) and (f) do not use $\kappa\ge\kappa_0$ or $\epsilon\le\epsilon_0$, so they hold for every $\kappa>0$ and $\epsilon>0$.
\end{proof}

Lemma \ref{lemma:interp} bounds the interpolation error $\|g_1 - G_1\|_\infty$ through the expected supremum of the increments of $G_1$ over pairs of points that are at most $1/\N$ apart in the sup norm.
\begin{lemma}\label{lemma:interp}
There exists a universal constant $A_0 > 0$ such that for all $\kappa > 0$, $\N \ge 1$ and $\epsilon > 0$ with
\begin{equation}\label{eqn: interp condition}
    \frac{A_0\,\kappa}{\N}\big\{\log (2\sqrt{2}\N)\big\}^{1/2} \le \frac{\epsilon}{2},
\end{equation}
we have
\begin{equation}\label{eqn: grid tail bound}
    \Pi\big(\|g_1 - G_1\|_\infty > \epsilon\big) \le 2\exp\Big(-\frac{\N^2\epsilon^2}{32\,\kappa^2}\Big).
\end{equation}
\end{lemma}
\begin{proof}
Let $\mc{P}_\N = \{(z,z') \in [0, 1]^{2} \times [0, 1]^2: \|z-z'\|_{\infty} \le 1/\N\}$ and $\omega_\N = \sup_{(z,z') \in \mc{P}_\N} |G_1(z) - G_1(z')|$. We show that $\|g_1 - G_1\|_\infty \le \omega_\N$, bound $\bbE\,\omega_\N$, and then apply Borell's inequality.

The products $\phi_i(z_1)\phi_j(z_2)$ are nonnegative, sum to one on $[0,1]^2$, and vanish unless $\|(z_1,z_2) - (u_i,u_j)\|_\infty < \Delta_\N$. Hence
\begin{equation*}
    \big\|g_1 - G_1\big\|_\infty = \bigg\|\sum^\N_{i,j=0} \big\{G_1(u_i,u_j) - G_1(\cdot)\big\} \phi_i(\cdot)\,\phi_j(\cdot) \bigg\|_\infty \le \omega_\N.
\end{equation*}

To bound $\bbE\,\omega_\N$, let $d_G(z,z') = \{\bbE(G_1(z) - G_1(z'))^2\}^{1/2}$ be the canonical metric of $G_1$. Since $1 - e^{-x} \le x$ for $x \ge 0$,
\begin{equation}\label{eqn: d and euclidean norm}
    d_G^2(z,z') = 2\big\{1-\exp(-\kappa^2\|z-z'\|_2^2)\big\} \leq 2\kappa^2 \|z-z'\|_2^2.
\end{equation}
Consider the centered Gaussian process $Y_{(z,z')} = G_1(z) - G_1(z')$ on $\mc{P}_\N$. Since $\mc{P}_\N$ is symmetric in $(z,z')$ and contains $(z,z)$, $\omega_\N = \sup_{\mc{P}_\N} Y_{(z,z')}$. Its canonical metric $d_Y$ satisfies $d_Y\big((z,z'),(\tilde z,\tilde z')\big) \le d_G(z,\tilde z) + d_G(z',\tilde z')$. By \eqref{eqn: d and euclidean norm} and $\|z-z'\|_2 \le \sqrt{2} \|z-z'\|_\infty$, we have $\{\bbE Y^2_{(z,z')}\}^{1/2} = d_G(z,z') \le 2\kappa/\N$ on $\mc{P}_\N$. So the $d_Y$ diameter of $\mc{P}_\N$ is at most $4\kappa/\N$.

Let $\mc{C}_Y(\epsilon)$ be the minimal number of $d_Y$ balls of radius $\epsilon$ covering $\mc{P}_\N$, and $\mc{C}_G(r)$ the minimal number of $d_G$ balls of radius $r$ covering $[0,1]^2$. By Corollary 4.15 of \cite{adler1990introduction}, there exists a universal constant $A_2$ such that
\begin{equation}\label{eqn: oscilation upper bound using integral}
\bbE\,\omega_\N \le A_2 \int_0^{\infty} \big(\log \mc{C}_Y(\epsilon)\big)^{1/2} d\epsilon = A_2 \int_0^{4\kappa/\N} \big(\log \mc{C}_Y(\epsilon)\big)^{1/2} d\epsilon,
\end{equation}
where the equality holds because $\mc{C}_Y(\epsilon) = 1$ for $\epsilon \ge 4\kappa/\N$. If $\{z_k\}$ are the centers of $\mc{C}_G(\epsilon/4)$ balls covering $[0,1]^2$, then each $d_Y$ ball of radius $\epsilon/2$ centered at $(z_k, z_l)$ that intersects $\mc{P}_\N$ is contained in a $d_Y$ ball of radius $\epsilon$ centered at a point of $\mc{P}_\N$. So $\mc{C}_Y(\epsilon) \le \mc{C}_G(\epsilon/4)^2$. By \eqref{eqn: d and euclidean norm}, a Euclidean ball of radius $r/(\sqrt{2}\kappa)$ is contained in a $d_G$ ball of radius $r$. Each square of side $r/\kappa$ lies in such a Euclidean ball, and $[0,1]^2$ is covered by at most $(1 + \kappa/r)^2$ of them. Hence $\mc{C}_G(r) \le (1 + \kappa/r)^2 \le (2\sqrt{2}\kappa/r)^2$ for $r \le \sqrt{2}\kappa$. Therefore, for $\epsilon \le 4\kappa/\N$,
\begin{equation*}
    \log \mc{C}_Y(\epsilon) \le 2\log \mc{C}_G(\epsilon/4) \le 4\log\frac{8\sqrt{2}\kappa}{\epsilon}.
\end{equation*}
For $A > 0$ and $0 < a \le A/e$, the substitution $\epsilon = Ae^{-y}$ and integration by parts give
\begin{equation}\label{eqn: incomplete gamma function form}
    \int_0^a \sqrt{\log\frac{A}{\epsilon}}\, d\epsilon = A\int_{x}^\infty y^{1/2}e^{-y}\, dy = A\Big\{x^{1/2}e^{-x} + \frac{1}{2}\int_x^\infty y^{-1/2}e^{-y}\, dy\Big\} \le \frac{3}{2}\,a\,\sqrt{\log\frac{A}{a}},
\end{equation}
where $x = \log(A/a) \ge 1$ and we used $\int_x^\infty y^{-1/2}e^{-y}\,dy \le x^{-1/2}e^{-x} \le x^{1/2}e^{-x}$. We apply \eqref{eqn: incomplete gamma function form} with $A = 8\sqrt{2}\kappa$ and $a = 4\kappa/\N$, for which $A/a = 2\sqrt{2}\N \ge e$. Then \eqref{eqn: oscilation upper bound using integral} yields
\begin{equation*}
\bbE\,\omega_\N \le 2A_2 \cdot \frac{3}{2}\cdot\frac{4\kappa}{\N}\big\{\log (2\sqrt{2}\N)\big\}^{1/2} = \frac{A_0\,\kappa}{\N}\big\{\log (2\sqrt{2}\N)\big\}^{1/2}
\end{equation*}
with $A_0 = 12A_2$.

Under \eqref{eqn: interp condition}, $\bbE\,\omega_\N \le \epsilon/2$. The variance of $Y_{(z,z')}$ on $\mc{P}_\N$ is at most $\sigma^2_\N := 4\kappa^2/\N^2$. By (2.6) in Theorem 2.1 of \cite{adler1990introduction},
\begin{equation*}
    \Pi(\omega_\N > \epsilon) \le 2\exp\Big\{-\frac{(\epsilon - \bbE\,\omega_\N)^2}{2\sigma^2_\N}\Big\} \le 2\exp\Big(-\frac{\N^2\epsilon^2}{32\,\kappa^2}\Big).
\end{equation*}
Since $\|g_1 - G_1\|_\infty \le \omega_\N$, this proves \eqref{eqn: grid tail bound}.
\end{proof}

\begin{lemma}\label{corr:sub_gaussian_1}
Suppose $Z\sim \cN(0,\Sigma)$ in $\bbR^d$, and let $\lambda_{\Sigma}$ and $\tr(\Sigma)$ denote the largest eigenvalue and the trace of $\Sigma$. Then for any $t > \sqrt{\tr(\Sigma)}$,
$$
P(\norm{Z}>t)\leq 2 \exp{\bigg(-\frac{\big(t - \sqrt{\tr(\Sigma)}\big)^2}{2\lambda_{\Sigma}}\bigg)}.
$$
\end{lemma}
\begin{proof}
 Consider the centered Gaussian process $X_u = u^\T Z$ indexed by the unit sphere $\{u \in \bbR^d: \norm{u} = 1\}$. Then $\sup_u X_u = \norm{Z}$, $\sup_u \bbE X^2_u = \sup_u u^\T\Sigma u = \lambda_\Sigma$, and $\bbE\norm{Z} \le (\bbE\norm{Z}^2)^{1/2} = \sqrt{\tr(\Sigma)}$. Hence, by (2.6) in Theorem 2.1 of \cite{adler1990introduction}, for $t > \sqrt{\tr(\Sigma)}$,
\begin{align*}
    P\big(\norm{Z}>t\big) \leq 2\exp\bigg(-\frac{\big(t - \bbE\norm{Z}\big)^2}{2\lambda_\Sigma}\bigg) \le 2 \exp{\bigg(-\frac{\big(t - \sqrt{\tr(\Sigma)}\big)^2}{2\lambda_{\Sigma}}\bigg)}.
\end{align*}
\end{proof}

\begin{lemma}\label{lemma:kappa_tail}
For all large $r$, $\Pi(\kappa_l > r) \le \exp\{-(D_2/2)\,r^2(\log r)^{4}\}$ $(l=1,2)$.
\end{lemma}
\begin{proof}
The density $\pi_\kappa$ satisfies (3.4) of \cite{van2009adaptive} with $d=2$ and $q=4$. By Lemma 4.9 of \cite{van2009adaptive}, $\Pi(\kappa_l > r) \le C_2 r^{p_\kappa-1}\exp\{-D_2r^2(\log r)^{4}\}/\{D_2(\log r)^{4}\}$ for all large $r$. The factor in front of the exponential is at most $\exp\{(D_2/2)r^2(\log r)^{4}\}$ for large $r$.
\end{proof}

\subsection{Choice of $\varepsilon_n$, $\kappa_n$ and $\N$}\label{ssec: choice of rates}
Let $\varepsilon_n = n^{-\frac{\alpha}{2\alpha + 2}} (\log n)^{\frac{2\alpha}{\alpha+1}}$ and $\kappa_n = c_\kappa\varepsilon_n^{-1/\alpha} = c_\kappa\, n^{\frac{1}{2\alpha + 2}} (\log n)^{-\frac{2}{\alpha+1}}$ with $c_\kappa = (10C_\alpha)^{1/\alpha}$. Here $C_\alpha$ and $A_1$ are the larger of the constants in Lemma \ref{lemma:gp} for $g^{(0)} = g^{(0)}_1$ and $g^{(0)} = g^{(0)}_2$. The constants $C_1, C_2, D_1, D_2$ and $p_\kappa$ are those of $\pi_\kappa$ in Theorem \ref{thm:post_cont}, and $\N\asymp n$.
The KL condition uses $\kappa_n \le \kappa \le 2\kappa_n$, and the sieve uses $1 \le \kappa \le c_r\kappa_n$ for a constant $c_r$ fixed in Section \ref{ssec: sieve condition}. The rest of the proof uses the following three properties, which hold for all large $n$.
\begin{itemize}
    \item[(P1)] $\varepsilon_n \ge 10\,C_\alpha\kappa^{-\alpha}$ for every $\kappa\ge\kappa_n$. This holds because $\kappa^{-\alpha} \le \kappa_n^{-\alpha} = c_\kappa^{-\alpha}\varepsilon_n$.
    \item[(P2)] $\kappa_n^2(\log n)^4 = c_\kappa^2\,n\varepsilon_n^2$, and hence $\kappa_n^2(\log n)^3 + \log n = o(n\varepsilon_n^2)$. The equality holds because $\kappa_n^2 = c_\kappa^2 n^{\frac{1}{\alpha+1}}(\log n)^{-\frac{4}{\alpha+1}}$ and $n\varepsilon_n^2 = n^{\frac{1}{\alpha+1}}(\log n)^{\frac{4\alpha}{\alpha+1}}$.
    \item[(P3)] $\N$ satisfies
    \begin{equation}\label{eqn: condition on N}
        \frac{A_0\,\kappa}{\N}\big\{\log (2\sqrt{2}\N)\big\}^{1/2} \le \frac{\varepsilon_n}{10}, \qquad \frac{\N^2\varepsilon_n^2}{800\,\kappa^2} \ge A_1\kappa^2\big\{\log(10\kappa/\varepsilon_n)\big\}^3 + \log 4\quad\text{for all } \kappa_n\le\kappa\le2\kappa_n.
    \end{equation}
    Both inequalities become tighter as $\kappa$ grows, so it suffices to check them at $\kappa=2\kappa_n$. The first then holds when $\N \gtrsim \kappa_n(\log n)^{1/2}/\varepsilon_n \asymp n^{1/2}(\log n)^{-3/2}$. The second holds when $\N \gtrsim \kappa_n^2(\log n)^{3/2}/\varepsilon_n \asymp n^{\frac{\alpha+2}{2\alpha + 2}}(\log n)^{-\frac{\alpha+5}{2\alpha+2}}$. Since $\frac{\alpha+2}{2\alpha+2}<1$ for $\alpha>0$, both hold under $\N\asymp n$.
\end{itemize}
Since $\kappa_n \to \infty$ and $\varepsilon_n \to 0$, Lemma \ref{lemma:gp} applies to every $\kappa\ge\kappa_n$ for large $n$. The quantities $\kappa_n/\varepsilon_n$, $\N$, and $Q_n$ and $\mc{T}_n$ below grow at most polynomially in $n$, so their logarithms are $\bigO(\log n)$.

\subsection{Sieve condition}\label{ssec: sieve condition}
Fix a constant $C > 0$. Let $r_n = c_r\kappa_n$, where the constant $c_r \ge 1$ satisfies $D_2c_r^2c_\kappa^2 \ge 2(2\alpha+3)^4C$. Define
\begin{equation*}
    B_n = Q_n r_n\mathbb{H}^{r_n}_1 + \frac{\varepsilon_n}{12}\mathbb{B}_1, \quad Q_n = -2\Phi^{-1}\big(e^{-Cn\varepsilon_n^2}\big), \quad \mc{T}_n = \sqrt{\tr(\Sigma)}\big(1 + \sqrt{2Cn\varepsilon_n^2}\big),
\end{equation*}
\begin{equation*}
    \Theta_n = \{(\xi^{(1)}, \xi^{(2)}, \beta): \xi^{(1)}, \xi^{(2)} \in T(B_n), \|\beta\|_2 \le \mc{T}_n \}, \qquad \mc{F}_n = \{f_\theta: \theta \in \Theta_n\}.
\end{equation*}
Since $\Phi(-t) \le e^{-t^2/2}$ for $t \ge 0$, we have $Q_n \le 2\sqrt{2Cn}\,\varepsilon_n$. By Lemma \ref{lemma:gp} (d), every $g \in B_n$ satisfies $\|g\|_\infty \le Q_nr_n + 1$. As $\|T(g)\|_\infty \le \|g\|_\infty$,
\begin{equation}\label{eqn: sieve xi bound}
    \|\xi^{(1)}\|_\infty \le Q_nr_n + 1 \quad \text{for all } (\xi^{(1)}, \xi^{(2)}, \beta) \in \Theta_n.
\end{equation}
Since $\theta \in \Theta_n$ implies $f_\theta \in \mc{F}_n$,
\begin{equation*}
    \Pi(\mc{F}^c_n) \le \Pi(\Theta_n^c) \leq \Pi(G_1 \notin B_n) + \Pi(G_2 \notin B_n) + \Pi(\|\beta\|_2>\mc{T}_n).
\end{equation*}
Since $\log r_n \ge \log\kappa_n \ge (\log n)/(2\alpha+3)$ for large $n$, Lemma \ref{lemma:kappa_tail}, (P2) and the choice of $c_r$ give $\Pi(\kappa_l > r_n) \le \exp\{-D_2c_r^2\kappa_n^2(\log n)^4/(2(2\alpha+3)^4)\} \le e^{-Cn\varepsilon_n^2}$. Let $1\le\kappa\le r_n$. By Lemma \ref{lemma:gp} (f), $Q_n\mathbb{H}^\kappa_1 \subset (Q_nr_n/\kappa)\mathbb{H}^{r_n}_1 \subset Q_nr_n\mathbb{H}^{r_n}_1$, so $\Pi(G^\kappa \notin B_n) \le \Pi(G^\kappa \notin Q_n\mathbb{H}^\kappa_1 + \frac{\varepsilon_n}{12}\mathbb{B}_1)$. By Lemma \ref{lemma:gp} (f) and (a) and by (P2), $-\log \Pi(\|G^\kappa\|_\infty < \varepsilon_n/12) \le -\log \Pi(\|G^{r_n}\|_\infty < \varepsilon_n/12) \le A_1c_r^2\kappa_n^2\{\log(12c_r\kappa_n/\varepsilon_n)\}^3 \le Cn\varepsilon_n^2$ for large $n$. Hence $\Phi^{-1}(\Pi(\|G^\kappa\|_\infty < \varepsilon_n/12)) \ge -Q_n/2$, and Lemma \ref{lemma:gp} (e) gives $\Pi(G^\kappa \notin B_n) \le 1 - \Phi(Q_n/2) = e^{-Cn\varepsilon_n^2}$. Therefore $\Pi(G_l \notin B_n) \le \Pi(\kappa_l > r_n) + \sup_{1\le\kappa\le r_n}\Pi(G^\kappa \notin B_n) \le 2e^{-Cn\varepsilon_n^2}$ for $l = 1, 2$. For $\beta \sim \cN(0, \Sigma)$ with $\Sigma$ in \eqref{eqn: covariance of beta}, $\lambda_\Sigma \le \tr(\Sigma)$ since $\Sigma$ is positive semidefinite, and Lemma \ref{corr:sub_gaussian_1} gives $\Pi(\|\beta\|_2 > \mc{T}_n) \le 2e^{-Cn\varepsilon_n^2}$. Therefore $\Pi(\mc{F}_n^c) \le 6e^{-Cn\varepsilon_n^2}$. Since $C$ is arbitrary, the same argument with $C+2$ in place of $C$ gives \eqref{eqn: sieve condition} for large $n$.

\subsection{Metric entropy condition}\label{ssec: entropy condition}
Let $\mc{G}_n$ be an $\varepsilon_n/6$ net of $B_n$ with respect to $\|\cdot\|_\infty$, and $\mc{M}_n$ an $\varepsilon_n/\{12\N(Q_nr_n+1)\}$ net of $\{\beta \in \bbR^{p\K}: \|\beta\|_2 \le \mc{T}_n\}$ with respect to $\|\cdot\|_2$. Since $\|T(g) - T(g')\|_\infty \le \|g - g'\|_\infty$, for any $\theta = (\xi^{(1)}, \xi^{(2)}, \beta) \in \Theta_n$ there exists $\tilde\theta = (T(\tilde g_1), T(\tilde g_2), \tilde\beta)$ with $\tilde g_1, \tilde g_2 \in \mc{G}_n$, $\tilde\beta \in \mc{M}_n$ and
\begin{equation*}
    \|\xi^{(1)} - T(\tilde g_1)\|_\infty \le \varepsilon_n/6, \quad \|\xi^{(2)} - T(\tilde g_2)\|_\infty \le \varepsilon_n/6, \quad \|\beta - \tilde\beta\|_2 \le \frac{\varepsilon_n}{12\N(Q_nr_n+1)}.
\end{equation*}
By Lemma \ref{lemma:hellinger} (iii) applied to $(\tilde\theta, \theta)$, and by \eqref{eqn: sieve xi bound}, $\dTheta(\tilde\theta, \theta) \le \varepsilon_n/2$ and thus $d_h(f_{\tilde\theta}, f_\theta) \le (\varepsilon_n/2)\, e^{1/4} < \varepsilon_n$. Using Corollary 4.2.11 of \cite{vershynin2026high} for the $\beta$ part,
\begin{equation*}
    \log \mc{C}(\varepsilon_n,\mc{F}_n, d_h) \le 2\log \mc{C}(\varepsilon_n/6, B_n, \|\cdot\|_\infty) + p\K \log\Big(1 + \frac{24\N(Q_nr_n+1)\mc{T}_n}{\varepsilon_n}\Big).
\end{equation*}
An $\varepsilon_n/12$ net of $Q_nr_n\mathbb{H}^{r_n}_1$ is an $\varepsilon_n/6$ net of $B_n$. So Lemma \ref{lemma:gp} (c) with $\kappa = r_n$ gives
\begin{equation*}
    \log \mc{C}(\varepsilon_n/6, B_n, \|\cdot\|_\infty) \le \log \mc{C}\big(\varepsilon_n/(12Q_nr_n), \mathbb{H}^{r_n}_1, \|\cdot\|_\infty\big) \le A_1\,r_n^2\big\{\log(12Q_nr_n/\varepsilon_n)\big\}^3.
\end{equation*}
Hence $\log \mc{C}(\varepsilon_n,\mc{F}_n, d_h) \lesssim c_r^2\kappa_n^2(\log n)^3 + \log n$, which is $o(n\varepsilon_n^2)$ by (P2). This gives \eqref{eqn: entropy condition} for large $n$.

\subsection{KL condition}\label{ssec: KL condition}
By Lemma \ref{lemma:hellinger} (ii), applied for each $(\bx, \s)$ and then integrated,
\begin{align}\label{eqn: KL sup bound}
    \Big\{\theta: \sup_{y, \bx, \s}\big|v_\theta - v_0\big| \le \varepsilon_n\Big\} \subset \Big\{\theta: \mrm{KL}(f_0,f_\theta)\leq C_0\varepsilon^2_n,\, \mrm{V}(f_0,f_\theta)\leq C_0\varepsilon^2_n\Big\}.
\end{align}
Hence \eqref{eqn: KL condition} with $A = C_0$ follows from a lower bound on the prior probability of the event on the left hand side.

Since $\|\bx\|_2\le1$ and $\|\beta^{(0)}(\s)\|_2\le1$, $w^{(0)}_{\bx,\s} \in [0,1]$, so $w^{(0)}_{\bx,\s} = c(w^{(0)}_{\bx,\s})$. Since $c$ is $1$-Lipschitz, $|w_{\bx,\s} - w^{(0)}_{\bx,\s}| \le \|\beta - \beta^{(0)}\|_2/2$. By Lemma \ref{lemma:hellinger} (i),
\begin{align*}
    \big|v_\theta - v_0\big| &\le \big|g_1(y, w_{\bx,\s}) - g_1(y, w^{(0)}_{\bx,\s})\big| + \big|g_1(y, w^{(0)}_{\bx,\s}) - g^{(0)}_1(y, w^{(0)}_{\bx,\s})\big| + \big|g_2(y, \s) - g^{(0)}_2(y, \s)\big| \\
    &\le 2\N\|\xi^{(1)}\|_\infty\|\beta - \beta^{(0)}\|_2 + \sum_{l=1}^2\Big(\big\|g_l - G_l\big\|_\infty + \big\|G_l - g^{(0)}_l\big\|_\infty\Big).
\end{align*}
Define the events
\begin{equation*}
    E_l = \Big\{\big\|G_l - g^{(0)}_l\big\|_\infty < \varepsilon_n/5,\ \big\|g_l - G_l\big\|_\infty \le \varepsilon_n/5\Big\} \ (l = 1, 2), \qquad
    E_\beta = \Big\{\|\beta - \beta^{(0)}\|_2 \le \frac{\varepsilon_n}{10\N(\|g^{(0)}_1\|_\infty + 1)}\Big\}.
\end{equation*}
On $E_1$, $\|\xi^{(1)}\|_\infty \le \|G_1\|_\infty \le \|g^{(0)}_1\|_\infty + 1$. So $\sup|v_\theta - v_0| \le \varepsilon_n$ on $E_1\cap E_2 \cap E_\beta$. Since $(G_1,\kappa_1)$, $(G_2,\kappa_2)$ and $\beta$ are independent a priori,
\begin{equation}\label{eqn: KL product}
    \Pi\Big(\mrm{KL}(f_0,f_\theta)\leq C_0\varepsilon^2_n,\, \mrm{V}(f_0,f_\theta)\leq C_0\varepsilon^2_n\Big) \ge \Pi(E_1)\,\Pi(E_2)\,\Pi(E_\beta).
\end{equation}

Since $\Sigma$ is positive definite, the density of $\cN(0,\Sigma)$ is bounded below by a positive constant on $\{\beta: \|\beta - \beta^{(0)}\|_2 \le 1\}$. The radius of $E_\beta$ decays polynomially in $n$. Hence $\Pi(E_\beta) \ge \exp(-C_\beta \log n)$ for some constant $C_\beta > 0$.

For $E_l$, $\Pi(E_l) \ge \Pi(\kappa_l\in[\kappa_n,2\kappa_n])\inf_{\kappa_n\le\kappa\le2\kappa_n}\Pi(E_l\mid\kappa_l=\kappa)$. Fix $\kappa\in[\kappa_n,2\kappa_n]$ and write $g^\kappa$ for the interpolation of $G^\kappa$. Then
\begin{align}
     \Pi(E_l\mid\kappa_l=\kappa) &\geq \Pi\big( \|G^\kappa - g^{(0)}_l\|_\infty < \varepsilon_n/5\big) - \Pi\big(\|g^\kappa - G^\kappa\|_\infty > \varepsilon_n/5\big) \label{eqn: g grid sup norm upper bound}.
\end{align}
By (P1) and Lemma \ref{lemma:gp} (b) with $\epsilon = \varepsilon_n/5$, the first term is at least $\exp[-A_1\kappa^2\{\log(5\kappa/\varepsilon_n)\}^3]$. By (P3) and Lemma \ref{lemma:interp} with $\epsilon = \varepsilon_n/5$, the second term is at most $2\exp\{-\N^2\varepsilon_n^2/(800\kappa^2)\} \le \frac{1}{2}\exp[-A_1\kappa^2\{\log(10\kappa/\varepsilon_n)\}^3]$. Hence 
\begin{equation*}
    \Pi(E_l\mid\kappa_l=\kappa) \ge \frac{1}{2}\exp[-A_1\kappa^2\{\log(10\kappa/\varepsilon_n)\}^3] \ge \frac{1}{2}\exp\{-4A_1\kappa_n^2(\log n)^3\},
\end{equation*}    
since $\kappa \le 2\kappa_n$ and $20\kappa_n/\varepsilon_n \le n$ for large $n$. For large $n$, $\pi_\kappa(\kappa) \ge C_1\kappa_n^{p_\kappa}\exp\{-4D_1\kappa_n^2(\log 2\kappa_n)^4\}$ on $[\kappa_n,2\kappa_n]$, so $\Pi(\kappa_l\in[\kappa_n,2\kappa_n]) \ge \exp\{-5D_1\kappa_n^2(\log n)^4\}$. Hence $\Pi(E_l) \ge \frac{1}{2}\exp\{-4A_1\kappa_n^2(\log n)^3 - 5D_1\kappa_n^2(\log n)^4\}$.

Plugging these bounds into \eqref{eqn: KL product} and using (P2),
\begin{align*}
    \Pi\Big(\mrm{KL}(f_0,f_\theta)\leq C_0\varepsilon^2_n,\, \mrm{V}(f_0,f_\theta)\leq C_0\varepsilon^2_n\Big) &\ge \frac{1}{4}\exp\Big[-8A_1\kappa_n^2(\log n)^3 - 10D_1\kappa_n^2(\log n)^4 - C_\beta\log n\Big] \\
    &\ge \exp\p{-C'n\varepsilon_n^2}
\end{align*}
with $C' = 10D_1c_\kappa^2 + 1$ and all large $n$. This verifies \eqref{eqn: KL condition} and completes the proof of Theorem \ref{thm:post_cont}.

\subsection{A growing number of covariates}\label{ssec: growing p}

Theorem \ref{thm:post_cont} fixes $p$ and $\K$. The number of covariates enters the proof only through $\beta\in\bbR^{p\K}$, and the corollary below lets $p=p_n$ grow with $n$. The covariates are bounded coordinatewise and the true index coefficients have unit $\ell_1$ norm, so the true index $\bx^\T\beta^{(0)}(\s)$ stays in $[-1,1]$ for every $p_n$, and the scale of the index is absorbed into $g^{(0)}_1$.

\begin{corollary}\label{cor:growing_p}
Let $p=p_n$ and let $\mc{H}_0$ be supported on $\{\bx\in\bbR^{p_n}:\|\bx\|_\infty\le1\}$. Suppose that $\log\Upsilon_0(y,\bx,\s)=g^{(0)}_1(y,\bx^\T\beta^{(0)}(\s))+g^{(0)}_2(y,\s)$, where $g^{(0)}_1$ is $\alpha$-H\"older on $[0,1]\times[-1,1]$, $g^{(0)}_2$ is $\alpha$-H\"older, neither depends on $n$, and $\max_{\s\in\mc{S}}\|\beta^{(0)}(\s)\|_1=1$. Under the prior of Theorem \ref{thm:post_cont} with $\K$ fixed, if $p_n\K\log n=o(n\varepsilon_n^2)$, then \eqref{eqn: posterior contraction} holds with the same $\varepsilon_n$.
\end{corollary}

Since $n\varepsilon_n^2=n^{\frac{1}{\alpha+1}}(\log n)^{\frac{4\alpha}{\alpha+1}}$, the condition is
\begin{equation*}
    p_n=o\big(n^{\frac{1}{\alpha+1}}(\log n)^{\frac{3\alpha-1}{\alpha+1}}\big),
\end{equation*}
for example $p_n=o\{n^{1/3}(\log n)^{5/3}\}$ when $\alpha=2$.

\begin{proof}
We follow the proof of Theorem \ref{thm:post_cont} and give only the steps that change. Since $\lvert\bx^\T\beta^{(0)}(\s)\rvert\le\|\bx\|_\infty\|\beta^{(0)}(\s)\|_1\le1$, the rescaling at the start of Supplement \ref{sec: proof of theorem} applies with $B=1$, so $w^{(0)}_{\bx,\s}=(\bx^\T\beta^{(0)}(\s)+1)/2\in[0,1]$, the rescaled $\tilde g^{(0)}_1(y,w)=g^{(0)}_1(y,2w-1)$ does not depend on $n$, and $\|\beta^{(0)}(\s)\|_2\le\|\beta^{(0)}(\s)\|_1\le1$. Hence Lemmas \ref{lemma:gp} and \ref{lemma:interp} apply with the same constants, and the choice of $\varepsilon_n$, $\kappa_n$ and $\N$ in Section \ref{ssec: choice of rates} is unchanged.

For $\|\bx\|_\infty\le1$, $\|\bx\|_2\le\sqrt{p_n}$, so $\lvert\tilde w-w\rvert\le\sqrt{p_n}\,\|\tilde\beta-\beta\|_2/2$ in the proof of Lemma \ref{lemma:hellinger} (iii). That lemma then holds with $\dTheta$ replaced by $\|\tilde\xi^{(1)}-\xi^{(1)}\|_\infty+\|\tilde\xi^{(2)}-\xi^{(2)}\|_\infty+2\N\sqrt{p_n}\,\|\xi^{(1)}\|_\infty\|\tilde\beta-\beta\|_2$, and the supremum is over $\|\bx\|_\infty\le1$.

The sieve condition holds as before, since Lemma \ref{corr:sub_gaussian_1} and $\lambda_\Sigma\le\tr(\Sigma)$ hold in any dimension. Here $\tr(\Sigma)=\sigma^2_\beta\K p_n$, and $p_n\le n$ for large $n$, so $\mc{T}_n$ still grows polynomially in $n$.

For the metric entropy condition, $\mc{M}_n$ is taken as an $\varepsilon_n/\{12\N\sqrt{p_n}(Q_nr_n+1)\}$ net, so that the modified $\dTheta(\tilde\theta,\theta)\le\varepsilon_n/2$ as before. The $\beta$ part of the bound becomes $p_n\K\log\{1+24\N\sqrt{p_n}(Q_nr_n+1)\mc{T}_n/\varepsilon_n\}$, and the logarithm in it is $\bigO(\log n)$. Hence $\log\mc{C}(\varepsilon_n,\mc{F}_n,d_h)\lesssim c_r^2\kappa_n^2(\log n)^3+p_n\K\log n$, which is $o(n\varepsilon_n^2)$ by (P2) and the assumption.

For the KL condition, $\lvert w_{\bx,\s}-w^{(0)}_{\bx,\s}\rvert\le\sqrt{p_n}\,\|\beta-\beta^{(0)}\|_2/2$, so the first term of the bound on $\lvert v_\theta-v_0\rvert$ becomes $2\N\sqrt{p_n}\|\xi^{(1)}\|_\infty\|\beta-\beta^{(0)}\|_2$, and $E_\beta$ is defined with the radius $b_n=\varepsilon_n/\{10\N\sqrt{p_n}(\|g^{(0)}_1\|_\infty+1)\}$. The covariance $\Sigma$ in \eqref{eqn: covariance of beta} is positive definite, so the RKHS of $\cN(0,\Sigma)$ on $(\bbR^{p_n\K},\|\cdot\|_2)$ is $\bbR^{p_n\K}$ with the norm $(h^\T\Sigma^{-1}h)^{1/2}$. Lemma 5.2 of \cite{van2008reproducing} with the symmetric set $\{b:\|b\|_2\le b_n\}$ gives
\begin{equation*}
    \Pi(E_\beta)\ge\exp\big(-\beta^{(0)\T}\Sigma^{-1}\beta^{(0)}/2\big)\,\Pi(\|\beta\|_2\le b_n).
\end{equation*}
The eigenvalues of $\Sigma=\sigma^2_\beta\Gamma_\beta\otimes\mathrm{I}_{p_n}$ are those of $\sigma^2_\beta\Gamma_\beta$, which do not depend on $n$ since $\K$ is fixed. Let $\lambda_{\min}$ and $\lambda_{\max}$ be the smallest and the largest. Then $\beta^{(0)\T}\Sigma^{-1}\beta^{(0)}\le\|\beta^{(0)}\|_2^2/\lambda_{\min}\le\K/\lambda_{\min}$. Write $\beta=\Sigma^{1/2}Z$ with $Z\sim\cN(0,\mathrm{I}_d)$ and $d=p_n\K$. Since $\|\beta\|_2\le\lambda_{\max}^{1/2}\|Z\|_2\le(\lambda_{\max}d)^{1/2}\|Z\|_\infty$, with $t_n=b_n/(\lambda_{\max}d)^{1/2}$,
\begin{equation*}
    \Pi(\|\beta\|_2\le b_n)\ge P(\|Z\|_\infty\le t_n)=\{2\Phi(t_n)-1\}^{d}\ge\{2\phi(1)\,t_n\}^{d},
\end{equation*}
where $\phi$ is the standard normal density, and the last step holds because $\phi\ge\phi(1)$ on $[-t_n,t_n]$ when $t_n\le1$, which is true for large $n$. Since $1/t_n$ grows polynomially in $n$, $-\log\Pi(E_\beta)\le C_\beta\,p_n\K\log n$ for a constant $C_\beta$ and all large $n$. This replaces the bound $\Pi(E_\beta)\ge\exp(-C_\beta\log n)$ of Section \ref{ssec: KL condition}, and since $p_n\K\log n=o(n\varepsilon_n^2)$, the final bound there still gives \eqref{eqn: KL condition} with $C'=10D_1c_\kappa^2+1$.
\end{proof}

\section{Details of the numerical experiments}\label{sec: numerical details}

\subsection{Benchmark methods}\label{ssec: benchmarks}

The full grid LGP models $f_\theta(y\m\bx,\s)\propto\exp\{g(y,\bx,\s)\}$ with a single function $g$. Each covariate is rescaled to $u_j=(x_j-a_j)/(b_j-a_j)\in[0,1]$, where $[a_j,b_j]$ is its range over all individuals, and $g$ is represented by linear interpolation on a tensor grid,
\begin{equation}\label{eqn: full grid}
    g(y,\bx,\s) = \sum_{i=0}^{\N}\sum_{r=0}^{\N_s}\sum_{\mathbf{t}\in\{0,\ldots,m-1\}^p}\zeta_{ir\mathbf{t}}\,\phi^{\N}_i(y)\,\phi^{\N_s}_r(\s)\prod_{j=1}^p\phi^{m-1}_{t_j}(u_j),
\end{equation}
where $\phi^{M}_i$ is the hat basis \eqref{eqn: hat basis} with $\N_l$ replaced by $M$. We assign $\zeta\sim\cN(0,\sigma^2\,\Gamma_y\otimes\Gamma_s\otimes\Gamma_x^{\otimes p})$, where each $\Gamma$ is the correlation matrix of a one dimensional Mat\'{e}rn kernel with $\nu=5/2$ and $\ell=0.292$ at the knots of its axis, and $\sigma^2\sim\mathrm{IG}(a_\sigma,b_\sigma)$. For fixed $(\bx,\s)$, the exponent is piecewise linear in $y$, so the normalizing constant is computed exactly as in Supplement \ref{ssec: normalizing constant}. The sampler alternates an ESS update of $\zeta$ and a conjugate update of $\sigma^2$, as for \ourmodel{}.

We use $\N=50$ and $\N_s=10$, with $m=3$ in Section~\ref{sec: simulation} and $m=2$ in Section~\ref{sec: real data}. Hence $\zeta$ has $(\N+1)(\N_s+1)m^p$ entries, $51\cdot11\cdot3^5=136{,}323$ for $p=5$ and $51\cdot11\cdot2^7=71{,}808$ for $p=7$, against $2(\N+1)^2=5{,}202$ for \ourmodel{}.

The spLGP of Section~\ref{sec: simulation} sets $\mathbf{z}_{i\s}=(\bx_i^\T,\s-1/2)^\T/c$, where $c$ is the largest norm of $(\bx_i^\T,\s-1/2)^\T$ over all individuals and locations, and models $f_\theta(y\m\bx_i,\s)\propto\exp\{g(y,w_{i\s})\}$ with $w_{i\s}=(\mathbf{z}_{i\s}^\T\beta+1)/2$. The function $g$ has the representation \eqref{eqn: basis expansion g} with $\N=50$, the Mat\'{e}rn kernel and the variance prior of \ourmodel{}, and $\beta\sim\cN(0,\mathrm{I}_{p+1})$ enters the likelihood through $\beta/\|\beta\|_2$. It is the sampler of Section~\ref{ssec: sampler} applied to the $n\K$ pairs $(\mathbf{z}_{i\s},y_{i\s})$ without $g_2$, under the same budget as \ourmodel{}.

LSBP \citep{rigon2021tractable} models $z=\mathrm{logit}(y)$, standardized by its mean and standard deviation over the training data, by the mixture $\sum_{h=1}^{H}\pi_h(\bx,\s)\,\cN\{z;\,\mathbf{a}^\T\gamma_h,\,\tau_h^{-1}\}$ with $\mathbf{a}=(1,\bx^\T,\s-1/2)^\T$ and $H=20$. The weights are $\pi_h=\nu_h\prod_{l<h}(1-\nu_l)$ with $\mathrm{logit}\,\nu_h=\mathbf{b}^\T\alpha_h$, where $\mathbf{b}$ holds $1$, $\bx$ and a natural cubic spline basis of $\s$ with three degrees of freedom, and $\nu_H=1$. Following the tutorial of the package, $\gamma_h\sim\cN(0,\mathrm{I})$, $\alpha_h\sim\cN(0,\mathrm{I})$ and $\tau_h\sim\mathrm{Gamma}(1,1)$. We run its Gibbs sampler for $2000$ burn-in and $2000$ posterior iterations, and transform each density back to the scale of $y$. Since a mixture has no piecewise linear exponent, its densities are evaluated on a grid of $2001$ points from every fourth posterior draw and interpolated linearly to the grid of Section~\ref{sec: simulation}.

\subsection{Coverage of the credible band}\label{ssec: coverage}

For each test individual $\bx$ and location $\s$, the coverage of the band is the probability, under $Y\sim f_0(\cdot\m\bx,\s)$, that the band contains $f_0(Y\m\bx,\s)$, and we average it over the test individuals and the locations. Averaged over all settings of Section~\ref{sec: simulation}, the coverage is $0.92$ and $0.87$ for \ourmodel{} in Case 1 and Case 2, against $0.84$ and $0.82$ for LSBP, $0.58$ and $0.55$ for the spLGP and $0.22$ and $0.30$ for the full grid LGP.

\subsection{Dependence across locations}\label{ssec: dependence}

Model \eqref{eqn: conditional density model} treats the responses of one individual as independent given $\bx_i$. To see how dependence along the tract affects the fit, we generate Case 1 \eqref{eqn: unimodal} with $n=250$ and $\K=20$ so that the responses of one individual are correlated across locations. For each individual $i$, we draw $(z_{i1},\ldots,z_{i\K})$ from a normal distribution with mean zero, unit variances and $\mathrm{corr}(z_{ik},z_{ik'})=\rho^{\lvert k-k'\rvert}$, and set $y_{i\s_k}=F_0^{-1}\{\Phi(z_{ik})\m\bx_i,\s_k\}$, where $\Phi$ is the standard normal CDF and $F_0(\cdot\m\bx,\s)$ is the CDF of \eqref{eqn: unimodal}. Each $y_{i\s_k}$ then still has the density \eqref{eqn: unimodal}, and $\rho=0$ gives the independent responses of Section~\ref{sec: simulation}. For each $\rho\in\{0,0.5,0.9\}$ we fit \ourmodel{} to $50$ datasets as in Section~\ref{sec: simulation}. Table \ref{tab:dependence} reports the error, the share of test responses in the central $80\%$ and $95\%$ intervals of $\widehat f$, and the coverage of the $95\%$ credible band of Supplement~\ref{ssec: coverage}.

\begin{table}[htbp!]
\centering
\caption{Case 1 with $n=250$ and $\K=20$ under dependence across locations, averaged over $50$ replications. Prediction coverage is the share of test responses in the central $80\%$ and $95\%$ intervals of $\widehat f$, and band coverage is that of the $95\%$ credible band.}
\label{tab:dependence}
\begin{tabular}{ccccc}
\hline
$\rho$ & Sq.\ Hellinger dist. & $80\%$ prediction & $95\%$ prediction & $95\%$ band \\
\hline
0 & 0.0080 & 0.803 & 0.954 & 0.89 \\
0.5 & 0.0104 & 0.792 & 0.953 & 0.82 \\
0.9 & 0.0165 & 0.787 & 0.950 & 0.69 \\
\hline
\end{tabular}
\end{table}

The predictive intervals keep their nominal coverage for every $\rho$, since the dependence does not change the conditional density at each location, which is the target of model \eqref{eqn: conditional density model}. The error grows with $\rho$ because correlated responses carry less information than independent ones, and the credible band becomes narrower than the actual uncertainty, since the posterior treats the $\K$ responses of an individual as independent.

\subsection{Posterior contrast in the ADNI data}\label{ssec: contrast}

Figure \ref{fig:contrast} shows the difference in the median of $f_\theta(\cdot\m\bx,\s)$ between an AD and a CN individual (AD minus CN) of Section~\ref{sec: real data}, whose other covariates are at their overall averages. The posterior median from all individuals is negative at all $83$ locations, from $-0.086$ to $-0.007$, in line with the known decline of FA with Alzheimer's disease \citep{sexton2011meta}. The fits to the five training sets follow the fit to all individuals, and only one of them changes sign, at three locations near location $70$.

\begin{figure}[htbp!]
\centering
\includegraphics[width=\textwidth]{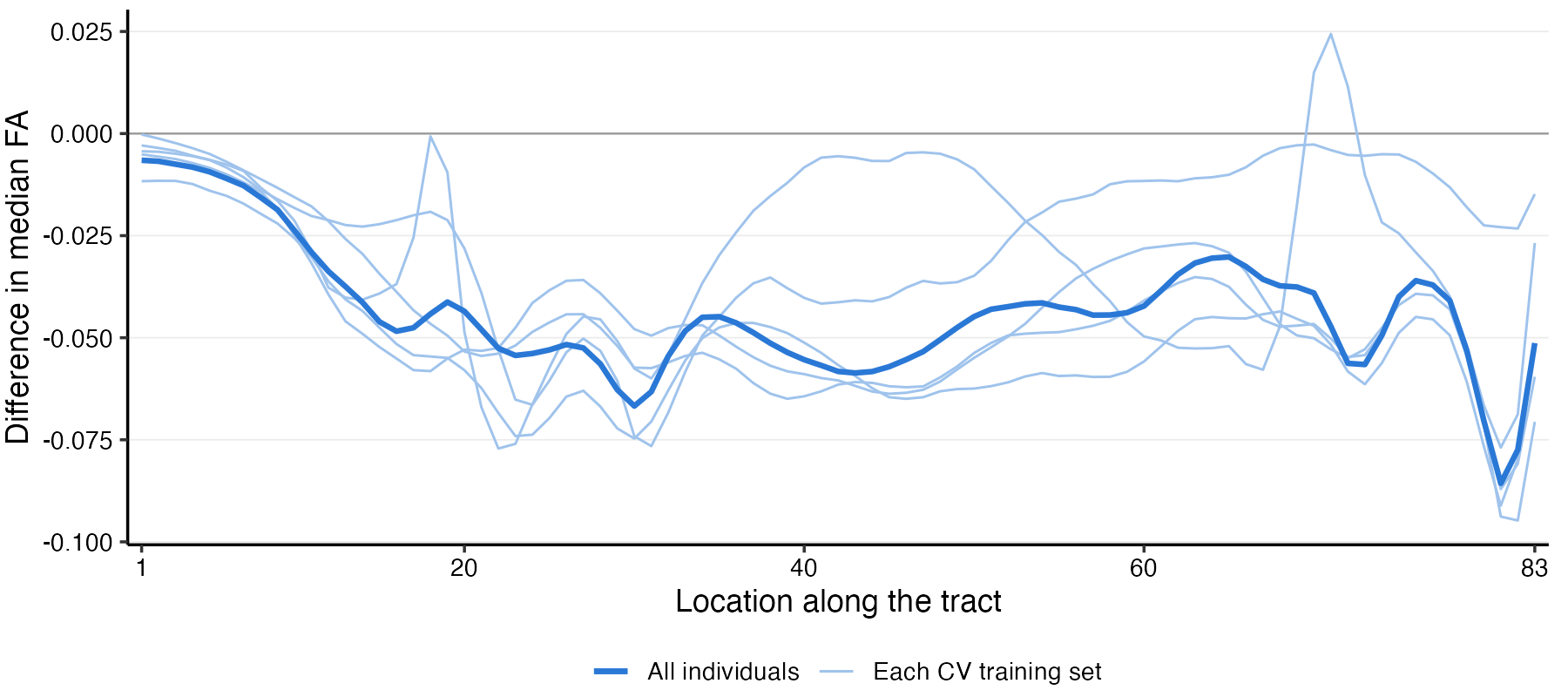}
\caption{Difference in the median of $f_\theta(\cdot\m\bx,\s)$ between an AD and a CN individual (AD minus CN) along the tract, with MMSE and the other covariates at their overall averages. The thick line is the posterior median from all $213$ individuals, and the thin lines are those from five training sets that each leave out one fifth of the individuals.}
\label{fig:contrast}
\end{figure}

\end{document}